\documentclass[11pt,a4paper]{article}
\usepackage[a4paper,margin=27mm]{geometry}
\usepackage[T1]{fontenc}
\usepackage{lmodern}
\usepackage{microtype}
\usepackage[authoryear]{natbib}
\usepackage{amsmath,amssymb,amsfonts,amsthm,mathtools}
\usepackage{booktabs,array,graphicx}
\usepackage{enumitem,caption}
\usepackage{longtable,pdflscape}
\usepackage{xcolor}
\usepackage[hidelinks]{hyperref}
\usepackage{setspace}

\newtheorem{assumption}{Assumption}[section]
\newtheorem{theorem}{Theorem}[section]
\newtheorem{remark}{Remark}[section]
\newtheorem{proposition}{Proposition}[section]
\newtheorem{lemma}{Lemma}[section]
\newtheorem{corollary}{Corollary}[section]
\newtheorem{definition}{Definition}[section]

\newcommand{\PP}{\mathbb{P}}
\newcommand{\EE}{\mathbb{E}}
\newcommand{\oP}{o_P}

\newcommand{\iid}{\overset{\mathrm{i.i.d.}}{\sim}}
\newcommand{\Ltwo}{L^2(P_0)}
\newcommand{\Psitrue}{\Psi(P_0)}
\newcommand{\Psihat}{\hat{\Psi}}

\newcommand{\Dstar}{D^{*}}

\newcommand{\ICj}{\mathrm{IC}_j}

\newcommand{\calP}{\mathcal{P}}

\newcommand{\norm}[1]{\left\|#1\right\|_{P_0}}
\newcommand{\Rtwo}{R_2}
\newcommand{\Rrem}{R_{\mathrm{rem}}}

\newcommand{\VarSand}{\widehat{\mathrm{Var}}_{\mathrm{Sand}}}
\newcommand{\VarJK}{\widehat{\mathrm{Var}}_{\mathrm{JK}}}
\newcommand{\VarBoot}{\widehat{\mathrm{Var}}_{\mathrm{Boot}}}
\newcommand{\VarHC}{\widehat{\mathrm{Var}}_{\mathrm{HC}}}
\newcommand{\sigIF}{\sigma^2_{\mathrm{IF}}}
\newcommand{\kappaM}{\kappa_{\mathrm{M}}}
\newcommand{\kappaG}{\kappa_{\mathrm{G}}}

\newcommand{\CIJK}{\mathrm{CI}^{\mathrm{JK}}_{1-\alpha}}

\newcommand{\PsiTMLE}{\hat{\Psi}_{\mathrm{TMLE}}}

\newtheorem{stheorem}{Theorem}[section]
\newtheorem{sproposition}[stheorem]{Proposition}

\newcommand{\expit}{\mathrm{expit}}

\begin{document}
\title{Variance Inference Beyond the Sandwich for Estimators with Second-Order Remainders}
\author{Lin Li$^{1}$ \qquad Pengcheng Wu$^{2}$\\[0.45em]
  \small $^{1}$Division of Biostatistics, University of California San Diego\\
  \small $^{2}$Department of Mechatronics, KTH Royal Institute of Technology\\[0.25em]
  \small \href{mailto:lil071@ucsd.edu}{lil071@ucsd.edu} \quad
  \href{mailto:penwu@kth.se}{penwu@kth.se}}
\date{November 2026}
\maketitle
\begin{abstract}
First-order semiparametric inference typically uses the same influence-function linearization to justify both centering and sandwich variance estimation. We show that these are distinct questions. In a near-boundary regime, the second-order remainder can be nearly mean zero yet large enough to contribute $O(n^{-1})$ sampling variation. When plug-in influence-function error and the cross term are negligible, the usual sandwich estimator captures only the influence-function component, and Wald intervals can undercover even when the point estimator is nearly unbiased. We derive a finite-sample variance decomposition and show that, when the cross term is asymptotically negligible, sandwich consistency is characterized by the condition $n\,\mathrm{Var}(\Rrem) \to 0$: the sandwich consistently estimates total variance in the strong-decay regime where this holds, and materially underestimates it in the near-boundary regime where $n\,\mathrm{Var}(\Rrem) \to c_R > 0$. We then study two corrections. Under a leave-one-out remainder linearization condition, the jackknife consistently estimates the total variance; under a conditional Mallows-2 bootstrap condition, the pairs bootstrap does likewise. For clustered data we distinguish within-cluster dependence from first-order influence-function mismatch. Simulations for a surrogate-assisted targeted learning estimator in stepped-wedge cluster-randomized trials illustrate variance discrepancies and the limits of jackknife correction.
\end{abstract}
\section{Introduction}
\label{sec:intro}

A common template in semiparametric statistics is that an estimator $\Psihat$ admits a von~Mises expansion,
\begin{equation}
  \Psihat - \Psitrue
  = \frac{1}{n}\sum_{i=1}^n \Dstar(O_i) + \Rrem,
  \label{eq:ale_intro}
\end{equation}
where $\Dstar$ is a reference influence function and $\Rrem$ is the residual, whose second-order structure requires additional conditions \citep{bickel1993efficient,newey1994asymptotic}. Write $\sigIF := \EE_{P_0}[D^*(O)^2]$ for the influence-function variance. Standard practice estimates the sampling variance of $\Psihat$ by $\sigIF/n$, yielding the familiar sandwich or Wald approximation. This reduction is justified when the remainder is negligible for \emph{variance} purposes, not merely negligible for point estimation.

The distinction matters near the product-rate boundary. In the strong-decay regime where $n\,\mathrm{Var}(\Rrem) \to 0$, the first-order linearization is sufficient for both asymptotic normality and first-order variance estimation, provided $\sqrt n\EE\Rrem\to0$ and the plug-in score is consistent. But many modern estimators are naturally analyzed only up to $\Rrem = O_P(n^{-1/2})$, the product-rate boundary: the remainder can have negligible mean while still large enough to contribute $O(n^{-1})$ sampling variation in the near-boundary regime. There a consistent influence-function sandwich estimator continues to target the influence-function component only, and confidence intervals based on it can undercover even when the bias of $\Psihat$ is negligible.

This paper studies that mismatch. Our focus is not on non-negligible \emph{bias}, which has been analyzed in other nonstandard semiparametric settings \citep{cattaneo2018kernel}, but on non-negligible \emph{remainder variance}. Related finite-sample concerns have appeared empirically for TMLE and doubly robust estimators \citep{vdl2011targeted,vermeulen2015bias}. We isolate a general mechanism behind those observations: the influence-function approximation can be adequate for the mean of the sampling distribution while still being inadequate for estimating its spread.

The analysis begins with a finite-sample variance decomposition separating the influence-function term, the remainder variance, and their covariance. When the cross term is negligible, $n\,\mathrm{Var}(\Rrem) \to 0$ characterizes whether a consistent influence-function sandwich captures the total variance. We also study jackknife and pairs-bootstrap variance estimation under, respectively, a leave-one-out remainder linearization condition and a conditional Mallows-2 condition. For clustered observations, we distinguish dependence within clusters from the effects of estimating an influence-function adjustment using the evaluation cluster itself.

The paper separates validity of first-order \emph{centering} from validity of first-order \emph{variance estimation}. In the strong-decay regime ($n\,\mathrm{Var}(\Rrem) \to 0$), the same first-order expansion justifies both when the bias and plug-in conditions above also hold. In the near-boundary regime, however, this equivalence can fail: the estimator can remain nearly unbiased while the usual sandwich misses an $O(n^{-1})$ contribution to the sampling variance.

Bias-correction work \citep[e.g.,][]{cattaneo2018kernel} addresses settings in which the center of the sampling distribution is itself non-negligibly shifted; here the center can already be first-order correct, while the failure occurs in the spread. Higher-order expansion and Edgeworth approaches seek a finer approximation to the full sampling law; here we isolate a specific $O(n^{-1})$ variance component omitted by the sandwich rather than pursuing a general higher-order approximation. Standard asymptotic-linearity results establish the expansion~\eqref{eq:ale_intro}; here that expansion is the starting point, and the question is when it is insufficient for variance estimation and under what conditions resampling recovers the missing contribution.

The results apply to semiparametric estimators with a von~Mises expansion and explicit moment, dependence, and resampling regularity conditions. The stepped-wedge application provides a setting in which the components of the variance discrepancy can be examined empirically.

\section{Setup and the Near-Boundary Regime}
\label{sec:setup}

\subsection{The von Mises expansion}

Let $\mathcal{P}$ be a nonparametric model and $\Psi: \calP \to \mathbb{R}$ a
pathwise differentiable functional with a mean-zero reference influence function $\Dstar
\in L^2_0(P_0)$. The reference may be the efficient influence function in the semiparametric model under study.  We study estimators $\Psihat$ that admit the von~Mises
expansion
\begin{equation}
  \Psihat - \Psitrue
  = \frac{1}{n}\sum_{i=1}^n \Dstar(O_i) + \Rrem,
  \label{eq:ale}
\end{equation}
where $\Rrem = \Rrem({\hat{P}}, P_0)$ is a second-order remainder satisfying
$\Rrem = O_P(n^{-1/2})$ (Assumption~\ref{ass:ale}).
We reserve ``asymptotically linear estimator'' (ALE) for the classical case
$\sqrt n\Rrem=o_P(1)$; the expansion considered here also allows larger residuals.
Throughout, we write
\[
  \bar D_n \;:=\; n^{-1}\sum_{i=1}^n D^*(O_i)
\]
for the influence-function average, and $\eta^0 = (\eta_1^0, \eta_2^0, \ldots)$
for the true nuisance functions appearing in the standard bilinear
representation of $\Rrem$ below.  When it is useful to make dependence on
nuisances explicit we write $D^*(O_i; \eta^0)$; in the unindexed notation
$D^*(O_i)$ this dependence is understood to be at the truth.
In the strong-decay
regime, with $\sqrt n\EE\Rrem\to0$, $\sqrt{n}\,\Rrem = \oP(1)$ and \eqref{eq:ale} reduces to classical
asymptotic linearity; in the near-boundary regime, $\sqrt{n}\,\Rrem = O_P(1)$
and the remainder is at the same stochastic order as the influence-function
term.  The key structural property of \eqref{eq:ale} is that $\Rrem$ is
\emph{second-order} in the nuisance estimation error under standard bilinear
structure (Definition~\ref{def:remainder}).

\begin{definition}[Standard von Mises Remainder Structure]
\label{def:remainder}
We say $\Rrem$ has \emph{standard bilinear structure} if it can be written as
\begin{equation}
  \Rrem
  = \int w_n\bigl(\hat{\eta}_1 - \eta^0_1\bigr)
         \bigl(\hat{\eta}_2 - \eta^0_2\bigr)\,dP_0
  + s_n,\qquad s_n=o_P(n^{-1/2}),\quad \|w_n\|_\infty\le C,
  \label{eq:bilinear}
\end{equation}
where $\eta_1, \eta_2$ are nuisance functionals with true values
$\eta_1^0, \eta_2^0$ under $P_0$ and $\hat{\eta}_k$ are their estimators.
For a targeted estimator, the population bilinear term is only part of the
residual: empirical-process and score-solving terms must also be included in
$s_n$ and controlled. By the Cauchy--Schwarz inequality,
$|\Rrem|\le C\norm{\hat\eta_1-\eta_1^0}\norm{\hat\eta_2-\eta_2^0}+|s_n|$,
so $\norm{\hat\eta_1-\eta_1^0}\norm{\hat\eta_2-\eta_2^0}=o_P(n^{-1/2})$
suffices for $\sqrt n\Rrem=o_P(1)$ under this representation.
\end{definition}

\begin{definition}[Product-Rate Boundary]
\label{def:boundary}
We say the estimator operates at the \emph{product-rate boundary} if each
nuisance error is of order $n^{-1/4}$ in $\Ltwo$:
\begin{equation}
  \norm{\hat{\eta}_k - \eta^0_k} = O_P(n^{-1/4}), \quad k=1,2.
  \label{eq:boundary}
\end{equation}
At the boundary the product $\norm{\hat{\eta}_1 - \eta^0_1} \cdot
\norm{\hat{\eta}_2 - \eta^0_2}$ is only $O_P(n^{-1/2})$, so the usual
sufficient condition $\oP(n^{-1/2})$ for asymptotic linearity need not hold.
Such rates can arise with growing-dimensional or nonparametric nuisance fits, but the rate bound alone does not establish a positive limiting remainder variance or a normal limit.
\end{definition}

\begin{remark}[Growing dimension]\label{rem:growing}
Dimension growth can also produce non-negligible second-order variance.
For independent Gaussian sample means $\bar U,\bar V\in\mathbb R^{p_n}$ with
coordinate variances $1/n$, $R_n=\bar U^\top\bar V$ has
$\EE R_n=0$ and $n\mathrm{Var}(R_n)=p_n/n\to\kappaG>0$.
The corresponding limit and its implications for the jackknife are derived in Supplement S10.
\end{remark}

\subsection{Regularity Conditions}

We work under the following conditions throughout Sections~\ref{sec:decomp}--\ref{sec:sand}.

\begin{assumption}[i.i.d.\ Data]
\label{ass:iid}
$O_1, \ldots, O_n \iid P_0$ with $O_i \in \mathcal{O}$. In a triangular model this is imposed within each row, with the mean and variance of $D^*$ equal to $0$ and $\sigIF$, and uniformly bounded fourth moments.
\end{assumption}

\begin{assumption}[Von Mises Expansion]
\label{ass:ale}
The estimator $\Psihat$ admits the expansion \eqref{eq:ale} with
$\Dstar \in L^2_0(P_0)$, $\sigIF = \EE_{P_0}[(\Dstar)^2] \in (0,\infty)$,
and square-integrable $\Rrem = O_P(n^{-1/2})$.
\end{assumption}

\begin{remark}[Relation to classical asymptotic linearity]
\label{rem:ale_relation}
Classical asymptotic linearity requires $\sqrt n\Rrem=o_P(1)$.
If $n\mathrm{Var}(\Rrem)\to0$ and $\sqrt n\EE\Rrem\to0$,
this follows by Chebyshev's inequality. At the variance boundary
$n\mathrm{Var}(\Rrem)\to c_R>0$, the residual may instead contribute
at the same scale as the influence-function average.
\end{remark}

\begin{assumption}[Influence-Function Moments]
\label{ass:bounded}
$\EE|\Dstar|^4<\infty$ (uniformly in a triangular model).
\end{assumption}

\begin{assumption}[Remainder Rate]
\label{ass:rate}
The limit $n\mathrm{Var}(\Rrem)\to c_R\ge0$ exists. We say $\Rrem$ is in
the \emph{strong-decay regime} if $c_R=0$, and in the \emph{near-boundary
variance regime} if $c_R>0$. The moment limit and any centering condition are imposed separately.
\end{assumption}

\section{Finite-Sample Variance Decomposition}
\label{sec:decomp}

\subsection{Variance decomposition}

\begin{theorem}[Finite-Sample Variance Decomposition]
\label{thm:decomp}
Under Assumptions~\ref{ass:iid}--\ref{ass:rate}, the sampling variance of
$\Psihat$ decomposes as:
\begin{equation}
  \mathrm{Var}\bigl(\Psihat - \Psitrue\bigr)
  \;=\;
  \underbrace{\frac{\sigIF}{n}}_{\substack{\text{influence-function}\\\text{variance}}}
  \;+\;
  \underbrace{\mathrm{Var}(\Rrem)}_{\substack{\text{remainder}\\\text{variance}}}
  \;+\;
  \underbrace{2\,\mathrm{Cov}\!\left(\frac{1}{n}\sum_{i=1}^n \Dstar(O_i),\;\Rrem\right)}_{\text{cross term}}.
  \label{eq:decomp}
\end{equation}
An oracle or consistent plug-in influence-function sandwich $\VarSand(\Psihat)$ targets only the
influence-function component $\sigIF/n$ and does not by itself account for
$\mathrm{Var}(\Rrem)$ or the cross term.
\end{theorem}

\noindent\emph{Proof.} See Supplement S5.1.\par

\subsection{The Near-Boundary Gap}
\label{sec:decomp:boundary}

\begin{proposition}[Variance Gap at the Product-Rate Boundary]
\label{prop:gap}
Suppose the remainder has standard bilinear structure (Definition~\ref{def:remainder})
with each nuisance at the product-rate boundary (Definition~\ref{def:boundary}), and suppose additionally that
$ n \cdot \mathrm{Cov}(\bar D_n,\Rrem) \to 0 $ and that the limit
$ c_R := \lim_{n\to\infty} n \cdot \mathrm{Var}(\Rrem) $
exists with $c_R > 0$.
Then:
\begin{enumerate}[label=\alph*.]
  \item $\mathrm{Var}(\Rrem) = c_R/n + o(n^{-1})$, so the total variance gap is
    \[
      \mathrm{Var}(\Psihat - \Psitrue) - \sigIF/n
      = c_R/n + o(n^{-1})
      = \Theta(n^{-1}).
    \]
  \item The relative gap satisfies
    \begin{equation}
      \frac{\mathrm{Var}(\Psihat - \Psitrue)}{\sigIF/n}
      \;=\; 1 + \frac{c_R}{\sigIF} + o(1),
      \label{eq:ratio}
    \end{equation}
    so the missing first-order variance component is asymptotically non-vanishing.
  \item The ratio \eqref{eq:ratio} is asymptotically \emph{constant}: it does
    not converge to 1 as $n\to\infty$ unless $c_R = 0$.
\end{enumerate}
\end{proposition}

\noindent\emph{Proof.} See Supplement S5.2.\par

\begin{remark}[Interpreting variance ratios]
\label{rem:constancy}
A non-unit empirical variance ratio across sample sizes indicates a variance discrepancy,
but it does not identify its source. Plug-in score error and first-order mismatch
can produce similar patterns; Section~\ref{sec:mismatch} gives an example.
\end{remark}

\begin{remark}[When the Cross Term Is Non-Negligible]
\label{rem:cross}
The cross term $2\,\mathrm{Cov}(\bar{D}_n, \Rrem)$ can be positive or
negative depending on the correlation structure between the influence-function
term and the remainder. In the SA-TMLE setting with CV-TMLE cross-fitting,
the out-of-fold construction does not by itself make $\bar D_n$ and
$\Rrem$ independent: training folds overlap and population remainders
depend on fitted nuisances. Thus $n\mathrm{Cov}(\bar D_n,\Rrem)\to0$
is a separate condition. It follows from Cauchy--Schwarz when
$n\mathrm{Var}(\Rrem)\to0$, but need not hold at the boundary.
\end{remark}

\section{Sandwich Consistency: A Sharp Characterization}
\label{sec:sand}

We now give a precise statement of when the sandwich is and is not consistent
for the total variance \citep{newey1994asymptotic}.
Here the oracle sandwich is the sample variance of $D^*(O_i)$ divided by $n$.
For an estimated score $\hat D_i$, the same limit requires, for example,
$n^{-1}\sum_i\{\hat D_i-D^*(O_i)\}^2=o_P(1)$.
This plug-in condition is assumed in Theorem~\ref{thm:sand}; it is not automatic
for non-vanishing nuisance errors.

\begin{theorem}[Sandwich Consistency Under Negligible Remainder Variance]
\label{thm:sand}
Under Assumptions~\ref{ass:iid}--\ref{ass:bounded}, suppose additionally that
$n\,\mathrm{Cov}(\bar D_n,\Rrem)\to 0$ and $n\VarSand\to_P\sigIF$.
\begin{enumerate}[label=\roman*.]
  \item If $n\cdot \mathrm{Var}(\Rrem) \to 0$, then
    \[
      \frac{\VarSand(\Psihat)}{\mathrm{Var}(\Psihat - \Psitrue)} \xrightarrow{P} 1.
    \]
  \item If $n\cdot \mathrm{Var}(\Rrem) \to c_R > 0$, then the sandwich targets
    only $\sigIF/n$ and underestimates the total variance by $c_R/n + o_P(n^{-1})$.
\end{enumerate}
\end{theorem}

\noindent\emph{Proof.} See Supplement S5.2.\par

\begin{corollary}[Standard Product-Rate Condition Does Not Suffice]
\label{cor:product_rate_insufficient}
Under the standard bilinear remainder structure of Definition~\ref{def:remainder},
the condition $\norm{\hat{\eta}_1 - \eta^0_1} \cdot
\norm{\hat{\eta}_2 - \eta^0_2} = \oP(n^{-1/2})$ implies $\sqrt{n}\,\Rrem =
\oP(1)$ (classically asymptotically linear) but does
\emph{not} by itself imply $n \cdot \mathrm{Var}(\Rrem) \to 0$ (sandwich consistency).
A sufficient strengthening is uniform integrability of $n\Rrem^2$.
Alternatively, under~\eqref{eq:bilinear}, it suffices that
\[
 n\EE\{\norm{\hat\eta_1-\eta_1^0}^2\norm{\hat\eta_2-\eta_2^0}^2\}\to0,
 \qquad n\EE s_n^2\to0.
\]
Faster $O_P$ rates alone are insufficient.
\end{corollary}

Variance consistency and limiting Wald coverage are distinct; see the rare-event example in Supplement S5.2.

\subsection{Foldwise plug-in second moments}\label{sec:plugin}
Use a partition fixed independently of the observations. Let $I_k$ be its fixed number $K$ of folds, $n_k=|I_k|\to\infty$, and $\pi_k=n_k/n$. For each $k$, $d_{n,k}$ is a random function measurable with respect to observations outside $I_k$. For a fresh $O\sim P_0$, define
\begin{align*}
 e_{n,k}&=d_{n,k}-D_n^*,& m_n&=\sum_k\pi_kP_0d_{n,k},\\
 \Lambda_n&=\sum_k\pi_kP_0e_{n,k}^2,& \Gamma_n&=\sum_k\pi_kP_0(D_n^*e_{n,k}).
\end{align*}
Here $D_n^*$ denotes the reference influence function in row $n$. Each population expectation holds that fold's fitted function fixed; we do not condition on the union of all training samples.
\begin{proposition}[Plug-in second-moment expansion]\label{prop:plugin}
Under Assumption~\ref{ass:iid}, suppose $\max_kP_0|d_{n,k}|^4=O_P(1)$. Then
\begin{equation}
 \frac1n\sum_k\sum_{i\in I_k}d_{n,k}(O_i)^2
 -\{\sigIF+\Lambda_n+2\Gamma_n\}\xrightarrow{P}0.
 \label{eq:plugin_limit}
\end{equation}
For the globally centered sample variance, the corresponding expression is $\sigIF+\Lambda_n+2\Gamma_n-m_n^2$. If $P_0e_{n,k}^2\xrightarrow{P}0$ for all $k$, the centered and uncentered limits both equal $\sigIF$. Otherwise a deterministic limit $\sigIF+\Lambda+2\Gamma-m^2$ requires the additional convergences $\Lambda_n\to_P\Lambda$, $\Gamma_n\to_P\Gamma$, and $m_n\to_Pm$; non-vanishing nuisance error alone implies none of them.
\end{proposition}
\noindent\emph{Proof.} See Supplement S7.\par
For globally targeted scores, extending this expansion requires control of the empirical difference between the final targeted scores and the foldwise fitted scores (Supplement S7).
\begin{remark}[Two discrepancies with unrestricted signs]\label{rem:two_errors}
When the displayed limits exist, $m=0$, and $c_\times:=\lim_n n\mathrm{Cov}(\bar D_n,\Rrem)$ exists, the plug-in-minus-total scaled variance gap is
\[
 \Lambda+2\Gamma-c_R-2c_\times.
\]
The sign depends on both plug-in score error and the covariance with the residual. Resampling consistency is governed by the conditions in Section~\ref{sec:refinements}.
\end{remark}

\section{Resampling-Based Variance Estimators}
\label{sec:refinements}

We consider leave-one-out jackknife and pairs-bootstrap estimators of the total sampling variance in~\eqref{eq:decomp}. Their consistency conditions differ, as do their computational requirements.

\subsection{The Leave-One-Out Jackknife}
\label{sec:jk}

\subsubsection{Definition and Rationale}

\begin{definition}[Leave-One-Unit-Out Estimator]
\label{def:loo}
For $i = 1,\ldots,n$, let $\hat{\Psi}^{(-i)}$ denote the estimator computed on the
sample with unit $i$ removed.  The jackknife variance estimator is:
\begin{equation}
  \VarJK(\Psihat)
  \;=\; \frac{n-1}{n}\sum_{i=1}^n
  \bigl(\hat{\Psi}^{(-i)} - \bar{\Psi}_{(-\cdot)}\bigr)^2,
  \label{eq:jk}
\end{equation}
where $\bar{\Psi}_{(-\cdot)} = n^{-1}\sum_i \hat{\Psi}^{(-i)}$.
\end{definition}

For cross-fitted estimators, deletion retains the fold labels, covariate dictionary,
and tuning rule of the original fit. This defines the fixed-partition jackknife
used in the simulations.

\subsubsection{Leave-One-Out Expansion}

\begin{lemma}[LOO Expansion]
\label{lem:loo}
Suppose Assumptions~\ref{ass:iid}--\ref{ass:bounded} hold.  For each $i$,
write $\Rrem^{(-i)}$ for the residual of the leave-one-out
estimator and define the remainder perturbation
\[
  b_i := \Rrem^{(-i)} - \Rrem,
  \qquad
  \bar b := n^{-1}\sum_{i=1}^n b_i.
\]
Then
\begin{equation}
  \hat{\Psi}^{(-i)} - \Psihat
  \,=\, -\frac{D^*(O_i;\eta^0)-\bar D_n}{n-1} + b_i
  \,=\, -\frac{1}{n}\Dstar(O_i) + b_i + a_i,
  \label{eq:loo_expand}
\end{equation}
where
\[
  a_i := \frac{n\bar D_n-D^*(O_i;\eta^0)}{n(n-1)}
\]
satisfies
\[
  \max_{1\le i\le n}|a_i| = O_P(n^{-1}),
  \qquad
  n^{-1}\sum_{i=1}^n a_i^2 = O_P(n^{-2}).
\]
\end{lemma}

\noindent\emph{Proof.} See Supplement S5.3.\par

\subsubsection{Jackknife Consistency}

Lemma~\ref{lem:loo} isolates the jackknife analysis into a negligible
algebraic correction $a_i$ and the genuine leave-one-out remainder perturbation
$b_i=\Rrem^{(-i)}-\Rrem$.  Whether the empirical quadratic sum
$(n-1)\sum_i(b_i-\bar b)^2$ concentrates around a limit depends on the
dependence and scaling of these leave-one-out remainder perturbations
\citep{shao1995jackknife}.  We state the required structure separately.

\begin{assumption}[LOO Remainder Linearization]
\label{ass:smooth}
The LOO perturbations $b_i = \Rrem^{(-i)}-\Rrem$ admit a uniform
linearization $\max_i|nb_i + \psi_{n,i}|\xrightarrow{P}0$ for mean-zero
$\{\psi_{n,i}\}$ satisfying: (i)~$n^{-1}\sum_i\psi_{n,i}=o_P(1)$;
(ii)~$n^{-1}\sum_i\psi_{n,i}^2\xrightarrow{P}c_R$;
(iii)~$n^{-1}\sum_i D^*(O_i;\eta^0)\,\psi_{n,i}=o_P(1)$;
(iv)~$\sup_n n^{-1}\EE[\max_i\psi_{n,i}^2]<\infty$
and $\sup_{n,i}\EE[\psi_{n,i}^4]<\infty$.
Additionally, each $\psi_{n,i}$ is a measurable function of $O_i$ alone
given the nuisance training data $\mathcal{T}_n$, and assume that
$\{\psi_{n,i}\}_{i=1}^n$ are conditionally independent given $\mathcal{T}_n$. The conditional independence condition concerns the chosen training construction.
\end{assumption}

The quadratic limit in Assumption~\ref{ass:smooth}(ii) is a condition on the
leave-one-out perturbations, not a consequence of the rate of $\Rrem$ alone.
For smooth fixed-dimensional plug-in models with bilinear remainder, the
first derivative of the population remainder is zero and the corresponding
limit is $c_R=0$. Supplement S1 treats regular nuisance fits, while
Supplement S10 gives a growing-dimensional example where the jackknife
quadratic limit differs from $c_R$.

\begin{proposition}[LOO Remainder Variance Stabilization]
\label{prop:loostab}
Under Assumption~\ref{ass:smooth}, in the near-boundary regime
($n\cdot\mathrm{Var}(\Rrem) \to c_R > 0$),
\begin{equation}
  C_n \;:=\; (n-1)\sum_{i=1}^n (b_i-\bar b)^2 \;\xrightarrow{P}\; c_R.
  \label{eq:delta_quad}
\end{equation}
\end{proposition}

\noindent\emph{Proof.} See Supplement S5.4.\par

\begin{theorem}[Jackknife Variance Consistency]
\label{thm:jk_consistency}
Under Assumptions~\ref{ass:iid}--\ref{ass:bounded}, let $b_i := \Rrem^{(-i)}-\Rrem$ and $\bar b := n^{-1}\sum_i b_i$.
\begin{enumerate}[label=\roman*.]
  \item \emph{(Strong decay regime, $n\cdot\mathrm{Var}(\Rrem) \to 0$):}
    if
    \[
      \max_i |b_i| = o_P(n^{-1/2})
      \qquad\text{and}\qquad
      \sum_i b_i^2 = o_P(n^{-1}),
    \]
    then $n \cdot \VarJK(\Psihat) \xrightarrow{P} \sigIF$.
  \item \emph{(Near-boundary regime, $n\cdot\mathrm{Var}(\Rrem) \to c_R > 0$):}
    if Assumption~\ref{ass:smooth} holds, then
    \[
      n \cdot \VarJK(\Psihat) \xrightarrow{P} \sigIF + c_R.
    \]
    When additionally $n\,\mathrm{Cov}(\bar D_n, \Rrem) \to 0$, the right-hand
    side coincides with the leading-order total sampling variance of
    Theorem~\ref{thm:decomp}, so the jackknife is consistent for the total
    variance.  Under the same condition and the plug-in consistency condition of Section~\ref{sec:sand}, the sandwich stays at $\sigIF$.
\end{enumerate}
\end{theorem}

\noindent\emph{Proof.} See Supplement S5.5.\par

\subsection{The Pairs Bootstrap}
\label{sec:boot}

\subsubsection{Definition}

\begin{definition}[Pairs Bootstrap Variance Estimator]
\label{def:boot}
Let $\{O_1^*,\ldots,O_n^*\}$ be a resample drawn with replacement from
$\{O_1,\ldots,O_n\}$.  The bootstrap replication $\hat{\Psi}^*$ is the
estimator computed on this resample.  After $B$ independent resamples:
\begin{equation}
  \VarBoot(\Psihat)
  \;=\; \frac{1}{B-1}\sum_{b=1}^B
  \bigl(\hat{\Psi}^{(b)} - \bar{\Psi}^*\bigr)^2,
  \quad
  \bar{\Psi}^* = B^{-1}\sum_b \hat{\Psi}^{(b)}.
  \label{eq:boot_var}
\end{equation}
\end{definition}

\subsubsection{Bootstrap Consistency}

We first state the bootstrap regularity condition explicitly.  Let
$T_n := \sqrt{n}(\Psihat - \Psitrue)$ and $T_n^* := \sqrt{n}(\hat\Psi^* -
\Psihat)$, where $\hat\Psi^*$ is the estimator recomputed on a pairs bootstrap
resample.  Let $\mathcal{D}_n = \sigma(O_1,\ldots,O_n)$ denote the data
$\sigma$-algebra and $d_2$ the Mallows-2 (Wasserstein-2) distance between
probability distributions on $\mathbb{R}$.

\begin{assumption}[Mallows-2 Bootstrap Consistency]
\label{ass:boot}
\hfill
\begin{enumerate}[label=\roman*.]
  \item $\sup_n \EE[T_n^2] < \infty$, and $\EE^*[|T_n^*|^2\mid\mathcal D_n]<\infty$ almost surely for every $n$.
  \item The bootstrap distribution of $T_n^*$ is distributionally and second-moment
    consistent for $T_n$ in the Mallows-2 sense:
    \begin{equation}
      d_2\!\Bigl(\mathcal{L}^*(T_n^* \mid \mathcal{D}_n),\;
                 \mathcal{L}(T_n)\Bigr)
      \;\xrightarrow{P}\; 0,
      \label{eq:mallows}
    \end{equation}
    where $\mathcal{L}^*(\cdot\mid\mathcal{D}_n)$ denotes the conditional law
    of $T_n^*$ given the data.
\end{enumerate}
\end{assumption}

The Mallows-2 condition controls both the bootstrap distribution and its
second moment \citep{bickel1981some,cheng2010bootstrap}. Standard nuisance
rate and entropy conditions can help verify it, but boundary-rate cases
require a separate conditional distributional argument.

\begin{theorem}[Bootstrap Variance Consistency]
\label{thm:boot_consistency}
Under Assumptions~\ref{ass:iid}--\ref{ass:rate} and
Assumption~\ref{ass:boot}, in the iterated limit $B\to\infty$ followed by $n\to\infty$:
\begin{equation}
  n\bigl\{\VarBoot(\Psihat)-\mathrm{Var}(\Psihat-\Psitrue)\bigr\}
  \;\xrightarrow{P}\;0.
  \label{eq:boot_consist}
\end{equation}
The limit is the sampling variance from Theorem~\ref{thm:decomp},
including the remainder variance $\mathrm{Var}(\Rrem)$ and the covariance
term $2\,\mathrm{Cov}(\bar{D}_n, \Rrem)$.  No LOO stability condition is
required. A joint sequence $B_n\to\infty$ also suffices if $\EE^*[|T_n^*|^4\mid\mathcal D_n]=O_P(1)$.
\end{theorem}

\noindent\emph{Proof.} See Supplement S2.\par

Unlike the jackknife, the bootstrap does not
require Assumption~\ref{ass:smooth}, which can fail when $n$ is very small or when ensemble
learning is highly nonlinear. It also produces a full distribution over
$\hat{\Psi}$, which can in principle support percentile-type or BCa
intervals when higher-order bootstrap validity is available \citep{efron1987better,hall1992bootstrap}. The consistency argument uses Assumption~\ref{ass:boot} (Mallows-2 consistency) rather than
the LOO expansion of Lemma~\ref{lem:loo} and Assumption~\ref{ass:smooth}.
The Mallows-2 condition is well-understood in the bootstrap literature and
can hold under standard moment and smoothness conditions for suitable learners,
although verification remains case-specific
\citep{bickel1981some,van_der_vaart1998asymptotic}.

\subsubsection{Bootstrap Confidence Intervals}

Beyond the Wald-type interval $\Psihat \pm z_{\alpha/2}\sqrt{\VarBoot(\Psihat)}$, the
bootstrap supports the bias-corrected and accelerated (BCa) interval of
\citet{efron1987better}.  Given $B$ bootstrap replicates
$\hat\Psi^{(1)},\ldots,\hat\Psi^{(B)}$, the BCa interval is
$(\hat\Psi^{*}_{[\alpha_1]},\hat\Psi^{*}_{[\alpha_2]})$, where the
percentile levels $\alpha_1, \alpha_2$ are obtained by adjusting
$\alpha/2, 1-\alpha/2$ using a bias-correction factor
$\hat z_0 = \Phi^{-1}\bigl(B^{-1}\#\{b : \hat\Psi^{(b)} \le \Psihat\}\bigr)$
and a jackknife-based acceleration estimate $\hat a$; see
\citet[p.~178]{efron1987better}.  Higher-order coverage for BCa requires
additional Edgeworth regularity beyond Assumption~\ref{ass:boot}
\citep{hall1992bootstrap}; BCa coverage is not implied by the variance-consistency result.

\subsection{Asymptotic Validity of Refined Confidence Intervals}
\label{sec:validity}

The variance consistency results of Theorems~\ref{thm:jk_consistency}
and~\ref{thm:boot_consistency} show that, under their respective assumptions and when the cross term is negligible, the jackknife and bootstrap both target the variance level $(\sigIF+c_R)/n$.  To convert variance
consistency into valid confidence intervals, we additionally need the
studentized pivot to converge in distribution to $N(0,1)$, which requires
asymptotic normality of the centered estimator $T_n :=
\sqrt{n}(\Psihat - \Psitrue)$.  The von~Mises expansion alone does not
guarantee this in the near-boundary regime, since $T_n = \sqrt{n}\bar{D}_n
+ \sqrt{n}\Rrem$ and $\sqrt{n}\Rrem = O_P(1)$ has non-negligible variance $c_R$
beyond the $N(0,\sigIF)$ limit of $\sqrt{n}\bar{D}_n$.

\begin{assumption}[Near-Boundary CLT]
\label{ass:clt}
$\sqrt{n}(\Psihat - \Psitrue) \;\Rightarrow\; N(0,\;\sigIF + c_R)$.
\end{assumption}

A sufficient condition for Assumption~\ref{ass:clt} is a joint Gaussian limit
of $(\sqrt n\bar D_n,\sqrt n\Rrem)$ with zero limiting covariance and
variances $\sigIF$ and $c_R$. Cross-fitting alone does not establish this
joint limit. Supplement S10 provides an explicit growing-dimensional
Gaussian example with $c_R>0$; the simulations in Section~\ref{sec:sim}
are evaluated separately at finite sample sizes.

\begin{theorem}[Bootstrap Wald Interval Validity]
\label{thm:boot_ci}
Under Assumptions~\ref{ass:iid}--\ref{ass:rate}, \ref{ass:boot},
and~\ref{ass:clt}, suppose additionally that
$n\,\mathrm{Cov}(\bar D_n,\Rrem) \to 0$.
With the Monte Carlo limit conditions of Theorem~\ref{thm:boot_consistency}, the studentized bootstrap pivot satisfies:
\begin{equation}
  T_n^{\rm boot} \;:=\; \frac{\Psihat - \Psitrue}{\sqrt{\VarBoot(\Psihat)}}
  \;\Rightarrow\; N(0,1).
\end{equation}
Consequently, the corresponding bootstrap Wald interval is asymptotically valid.
\end{theorem}

\begin{proof}
By Theorem~\ref{thm:boot_consistency},
$n\VarBoot(\Psihat)\xrightarrow{P}\sigIF+c_R>0$.
Assumption~\ref{ass:clt} and Slutsky's theorem yield the stated limit.
\end{proof}

\begin{theorem}[Jackknife Wald Interval Validity]
\label{thm:jk_ci}
Under the conditions of Theorem~\ref{thm:jk_consistency}, Assumption~\ref{ass:clt},
and $n\mathrm{Cov}(\bar D_n,\Rrem)\to0$, the studentized jackknife pivot satisfies:
\begin{equation}
  T_n^{\rm JK} \;:=\; \frac{\Psihat - \Psitrue}{\sqrt{\VarJK(\Psihat)}}
  \;\Rightarrow\; N(0,1).
  \label{eq:pivot}
\end{equation}
Consequently, $\CIJK = \Psihat \pm z_{\alpha/2}\sqrt{\VarJK(\Psihat)}$
is asymptotically valid: $P(\Psitrue \in \CIJK) \to 1-\alpha$.
\end{theorem}

\noindent\emph{Proof.} See Supplement S3.\par

\subsection{Variance ratio representation}
\label{sec:hc}

It is sometimes useful to re-express the jackknife as a multiplicative
inflation of the sandwich rather than as a separate variance estimator.

\begin{definition}[HC-Corrected Sandwich]
\label{def:hc}
When $\VarSand(\Psihat)>0$, define the estimated variance ratio by
\[\hat{\rho}=\VarJK(\Psihat)/\VarSand(\Psihat).\]
The HC-corrected sandwich is:
\begin{equation}
  \VarHC(\Psihat) \;=\; \hat{\rho} \cdot \VarSand(\Psihat).
  \label{eq:hc}
\end{equation}
The corresponding confidence interval is $\Psihat \pm z_{\alpha/2}
\sqrt{\VarHC(\Psihat)}$.
\end{definition}

By construction, $\VarHC \equiv \VarJK$, so the HC-corrected interval is
numerically identical to the jackknife Wald interval.  The computational
cost is also the same: $\hat\rho$ is obtained from jackknife refits.  The
value of the representation is conceptual, not computational: it makes
explicit that the near-boundary correction is a \emph{scaling} of the
sandwich rather than a fundamentally different object, and it exposes
$\hat{\rho}$ as an estimator of $1+c_R/\sigIF$ only when the jackknife and sandwich limits and the negligible-cross-term condition have been verified. Without those conditions it is a descriptive variance ratio.

\subsection{Comparison of resampling procedures}

\begin{remark}[First-Order Equivalence]
\label{rem:equiv}
Under the total-variance consistency conditions of Theorems~\ref{thm:jk_consistency} and~\ref{thm:boot_consistency},
$\VarJK = \VarHC$ and $\VarJK/\VarBoot \xrightarrow{P} 1$:
the three estimators are first-order equivalent.
The cluster jackknife requires $J$ refits, and the simulations use $t_{J-1}$ critical values;
the bootstrap requires $B$ bootstrap refits and supports percentile-type or
BCa intervals when higher-order bootstrap validity is available; the
HC-corrected sandwich is numerically identical to the jackknife and shares
its computational cost.
Bootstrap confidence intervals require the additional distributional conditions of Section~\ref{sec:validity}.
\end{remark}

\section{Extension to Clustered Data}
\label{sec:cluster}

The framework of Sections~\ref{sec:setup}--\ref{sec:refinements} extends to
clustered data by replacing $O_i$ with the cluster-level vector
$O_j = \{O_{ijt}\}_{i,t}$ and $n$ with $J$.  The expansion becomes
$\PsiTMLE-\Psitrue=J^{-1}\sum_{j=1}^J\ICj+\Rrem$,
where $\ICj=(mT)^{-1}\sum_{i,t}\Dstar(O_{ijt})$ is the normalized
cluster contribution for an individual-average target. In the simulations the reference contribution comes from the marginal observed-data model. The total residual $\Rrem$ includes terms beyond the population bilinear component $\Rtwo$.

\subsection{Within-cluster dependence}
\begin{proposition}[Conditional variance of a cluster average]\label{prop:icc}
Let $\mathcal T$ be an external training sample, independent of $J$ i.i.d. evaluation clusters. Conditional on $\mathcal T$, suppose a fixed learned score produces contributions $r_{j\ell}$, $\ell=1,\ldots,M$, exchangeable within a cluster conditional on $\mathcal T$, with conditional mean $\mu_{\mathcal T}$, variance $v_{\mathcal T}$, and pairwise covariance $c_{\mathcal T}$. For
$R=J^{-1}\sum_j M^{-1}\sum_\ell r_{j\ell}$,
\begin{equation}
 \mathrm{Var}(R\mid\mathcal T)=\frac{v_{\mathcal T}+(M-1)c_{\mathcal T}}{JM},\qquad
 \mathrm{Var}(R)=\frac{\EE[v_{\mathcal T}+(M-1)c_{\mathcal T}]}{JM}
       +\mathrm{Var}(\mu_{\mathcal T}).
 \label{eq:icc_amp}
\end{equation}
\end{proposition}
\begin{proof}
Expand the variance within each cluster, use conditional independence between clusters, and apply the law of total variance. See Supplement S6 for the non-exchangeable version.
\end{proof}
With $v_{\mathcal T}$ held fixed, positive $c_{\mathcal T}$ increases the conditional variance relative to independent within-cluster contributions. For fixed $v$ and $c\ge0$, this variance equals $c/J+(v-c)/(JM)$ and approaches $c/J$ as $M$ grows. The result conditions on external training; overlapping training samples in ordinary cross-fitting require a different dependence calculation.

\subsection{Clustered Jackknife and Bootstrap}

In the clustered setting, the LOO jackknife deletes entire clusters to
preserve within-cluster dependence.
For $j = 1,\ldots,J$, let $\hat{\Psi}^{(-j)}$ be the SA-TMLE on
$\mathcal{D}^{(-j)} = \{O_k : k \neq j\}$; the cluster jackknife variance is
$\VarJK(\PsiTMLE) = \frac{J-1}{J}\sum_{j=1}^J
(\hat{\Psi}^{(-j)} - \bar{\Psi}_{(-\cdot)})^2$.
The analogues of Section~\ref{sec:jk} extend to the cluster level, with
$n$ replaced by $J$ and $O_i$ replaced by $O_j$, under the corresponding
cluster-level LOO stability and bootstrap moment conditions.
Independence of $\{O_j\}$ does not itself verify conditions on fitted nuisances.  The cluster bootstrap resamples entire clusters with replacement.
In the simulations, the jackknife deletes whole clusters while retaining fold labels and the covariate dictionary.

\subsection{Own-cluster adjustment and first-order mismatch}
\label{sec:mismatch}

Clustered
designs also invite a different and more consequential practice: letting the
outcome model carry a cluster-specific intercept (a random effect, a cluster
fixed effect, or simply the cluster's mean residual) that is estimated from
the cluster's own data, and then evaluating the influence function at those
fitted values.  Cross-fitting does not remove the dependence, because the
adjustment for cluster $j$ is computed from cluster $j$.  We show in a
stylized model that this changes the estimator's influence function at first
order, and that the plug-in sandwich then converges to a fixed fraction of
the true variance.  The fraction does not depend on $J$, so the ratio
diagnostic of Remark~\ref{rem:constancy} cannot tell this regime from the
near-boundary one. 

\paragraph{Stylized model (M).}
Clusters $j=1,\ldots,J$ are i.i.d.  Cluster $j$ crosses over at a random
period $\tau_j$ with known marginal $g_t := P(\tau_j\le t)\in(0,1)$ for
$t=1,\ldots,T$, so that $A_{jt}=\mathbf 1\{\tau_j\le t\}$ is monotone in
$t$; a cluster effect $b_j$ with mean zero and variance $\sigma_b^2$ is
independent of $\tau_j$; and for $i=1,\ldots,m$ new individuals in each
period,
\[
  Y_{ijt} \;=\; \beta A_{jt} + \gamma_t + b_j + \varepsilon_{ijt},
  \qquad \EE[\varepsilon_{ijt}]=0,\quad\mathrm{Var}(\varepsilon_{ijt})=\sigma_\varepsilon^2.
\]
The errors $\varepsilon_{ijt}$ are i.i.d. across individuals, clusters, and periods,
independent of $(\tau_j,b_j)$, with
$\EE[b_j^4]+\EE[\varepsilon_{ijt}^4]<\infty$ and
$\sigma_b^2+\sigma_\varepsilon^2>0$.
The target is $\Psitrue=\beta$. The symbol $\sigIF$ denotes the variance of the marginal reference score in model~(M). Write
$h_{jt} := A_{jt}/g_t-(1-A_{jt})/(1-g_t)$, $\bar h_j := T^{-1}\sum_t h_{jt}$,
$\bar A_j := T^{-1}\sum_t A_{jt}$, $c_g := \sum_t g_t(1-g_t)$,
$\tilde a_{jt} := (A_{jt}-g_t)/c_g$, $a_j := \sum_t \tilde a_{jt}$ and
$\kappaM := \EE[\bar h_j \bar A_j]$.  Since $A_{jt}$ is monotone,
$\mathrm{Cov}(A_{jt},A_{js})=g_{t\wedge s}(1-g_{t\vee s})>0$, which gives
$\kappaM>0$ and $\EE[a_j^2]>0$.

\paragraph{Estimators.}
The clusters are partitioned into $V\ge 2$ folds of sizes $J_v$ with
$|J_v-J/V|\le 1$, and $v(j)$ denotes the fold of cluster $j$.  On the
training folds, $\hat Q^{(-v)}(a,t)=\hat\alpha+\hat\beta a+\hat\gamma_t$ is
the least-squares fit of $Y$ on an intercept, $A$ and period indicators. To define the fit on every finite sample, replace the normalized residual-treatment sum of squares $\sum_t\bar A_t^{-v}(1-\bar A_t^{-v})$ by its maximum with $c_g/2$ in the treatment-coefficient formula; then fit the period means conditional on that coefficient. This stabilization is inactive with probability tending to one exponentially fast and also defines all deletion fits.
With $N=JmT$, the cross-fitted estimator with known $g$ is
\begin{equation}
  \Psihat \;=\; \frac1N\sum_{j,i,t}\Bigl[\hat Q^{(-v(j))}(1,t)-\hat Q^{(-v(j))}(0,t)
  + h_{jt}\bigl\{Y_{ijt}-\hat Q^{(-v(j))}(A_{jt},t)\bigr\}\Bigr],
  \label{eq:base_est}
\end{equation}
and the own-cluster version $\Psihat_{\mathrm{own}}$ replaces
$\hat Q^{(-v(j))}$ by $\tilde Q_j := \hat Q^{(-v(j))}+\bar r_j$, where
\[
  \bar r_j := \frac{1}{mT}\sum_{i,t}\bigl\{Y_{ijt}-\hat Q^{(-v(j))}(A_{jt},t)\bigr\}
\]
is cluster $j$'s mean residual; equivalently
$\Psihat_{\mathrm{own}} = \Psihat - J^{-1}\sum_j \bar h_j\bar r_j$.
For either estimator the plug-in cluster sandwich is
$\VarSand = \{J(J-1)\}^{-1}\sum_j(\widehat{\ICj}-\overline{\widehat{\mathrm{IC}}})^2$
with
\[
  \widehat{\ICj} = \frac{1}{mT}\sum_{i,t}\Bigl[\,Q(1,t)-Q(0,t)-\Psihat
  + h_{jt}\{Y_{ijt}-Q(A_{jt},t)\}\Bigr]
\]
evaluated at the fitted $Q$ and the corresponding point estimate.

\begin{theorem}[Own-cluster estimator and plug-in limits]
\label{thm:own}
Under model~(M), as $J\to\infty$ with $m$, $T$ and $V$ fixed:
\begin{enumerate}[label=\roman*.]
  \item \emph{(Cross-fitted estimator.)}
    $\Psihat=\beta+J^{-1}\sum_j D_j+O_P(J^{-1})$ with
    the reference cluster contribution
    \[
      D_j := \bar h_j b_j+\frac1{mT}\sum_{i,t}h_{jt}\varepsilon_{ijt}.
    \]
    Hence
    \[
      J\,\mathrm{Var}(\Psihat)\;\to\;\sigIF
      := \sigma_b^2\,\EE[\bar h_j^2]+\frac{\sigma_\varepsilon^2\,\EE[h^2]}{mT},
    \]
    where $\EE[h^2]=T^{-1}\sum_t\{g_t(1-g_t)\}^{-1}$,
    and $J\,\VarSand(\Psihat)\xrightarrow{P}\sigIF$.
  \item \emph{(Own-cluster estimator.)}
    $\Psihat_{\mathrm{own}}=\beta+J^{-1}\sum_j\phi_j+O_P(J^{-1})$ with
    \[
      \phi_j \;:=\; \kappaM\,a_j b_j+\frac{1}{mT}\sum_{i,t}c_{jt}\,\varepsilon_{ijt},
      \qquad
      c_{jt} \;:=\; h_{jt}-\bar h_j+\kappaM T\tilde a_{jt},
    \]
    so that $\sqrt J(\Psihat_{\mathrm{own}}-\beta)\Rightarrow N(0,\sigma^2_{\mathrm{own}})$ with
    \[
      \sigma^2_{\mathrm{own}} := \kappaM^2\,\EE[a_j^2]\,\sigma_b^2
      + \frac{\sigma_\varepsilon^2\,\EE[c^2]}{mT},
      \qquad \EE[c^2]:=\frac1T\sum_t\EE[c_{jt}^2].
    \]
    The reference and own-cluster influence functions differ whenever $\EE(\phi_j-D_j)^2>0$.
  \item \emph{(Plug-in sandwich of the own-cluster estimator.)}
    \[
      J\,\VarSand(\Psihat_{\mathrm{own}})\;\xrightarrow{P}\;\sigma^2_{\mathrm{sand}}
      := \frac{\sigma_\varepsilon^2\,\{\EE[h^2]-\EE[\bar h_j^2]\}}{mT}.
    \]
    Consequently
    \[
      \frac{\VarSand(\Psihat_{\mathrm{own}})}{\mathrm{Var}(\Psihat_{\mathrm{own}})}
      \;\xrightarrow{P}\; \rho_{\mathrm{own}}:=\frac{\sigma^2_{\mathrm{sand}}}{\sigma^2_{\mathrm{own}}},
    \]
    a constant that does not depend on $J$. If $\sigma_b^2>0$, then
    \[
      \rho_{\mathrm{own}}\le
      \frac{\sigma_\varepsilon^2\,\{\EE[h^2]-\EE[\bar h_j^2]\}}{mT\,\kappaM^2\,\EE[a_j^2]\,\sigma_b^2}\to0
      \quad(m\to\infty\text{ after }J\to\infty).
    \]
  \item \emph{(Decomposition.)}
    In Theorem~\ref{thm:decomp} applied to $\Psihat_{\mathrm{own}}$ with
    $\Dstar$ the influence function of part~(i),
    $J\,\mathrm{Var}(\Rrem)\to\EE[(\phi_j-D_j)^2]$ and
    $2J\,\mathrm{Cov}(\bar D_J,\Rrem)\to 2\,\EE[D_j(\phi_j-D_j)]$.  The gap
    between $\sigIF/J$ and the true variance is therefore represented by both residual variance and covariance. Its first-order origin follows from the explicit expression $\phi_j-D_j$, not from their signs.
\end{enumerate}
\end{theorem}

\noindent\emph{Proof.} See Supplement S8.\par

\begin{remark}[Numerical illustration]
\label{rem:own_diag}
For $T=5$, $g_t=t/6$ ($t=1,\ldots,5$), $m=40$, $\sigma_b^2=0.0337$ and $\sigma_\varepsilon^2=0.64$, exact enumeration of the six equally likely crossover sequences gives
\[
 \kappaM=0.60,\quad \sigIF=0.122653,\quad
 \sigma^2_{\mathrm{own}}=0.058814,\quad
 \sigma^2_{\mathrm{sand}}=0.007555,\quad \rho_{\mathrm{own}}=0.12845.
\]
The residual variance and twice its covariance with $D_j$ contribute $0.021824$ and $-0.085662$, respectively. Within model~(M), the fixed-partition jackknife and a sandwich based on a consistent estimate of $\phi_j$ are consistent (Supplement S8).
\end{remark}

\section{Simulation Illustration}\label{sec:sim}
\subsection{Design and Data-Generating Process}
We use the stepped-wedge generator motivated by \citet{li2026core}, with $J$ independent clusters, seven periods, and $m=40$ new individuals per cluster-period. Each cluster crosses over at $\tau_j$, independently uniform on $\{2,\ldots,7\}$. Only periods $t=2,\ldots,6$ enter estimation, with known treatment probabilities $g_t=(t-1)/6$. This excludes the non-overlap endpoints. The target is the equal-weight average marginal contrast over these five periods, $\Psi_0=0.12$.

Five-fold cross-fitting splits whole clusters. The initial outcome regression is linear in $A,S,AS,W_1,W_2,W_3$ and period indicators. The integrated outcome regression uses $A,W_1,W_2,W_3$ and period indicators. The censoring model fits the observed-data Bernoulli likelihood with probability
\[
 g_\Delta(A,S,W,t)=a_0(A,t)\operatorname{expit}(x^\top\beta),
\]
with $a_0(A,t)=0.8$ for $A=1,t=6$ and one otherwise,
where $x=(1,S,A,W_3,\mathbf1\{t=6\})$. The known factor $a_0$ accounts for administrative censoring. A fixed ridge penalty of one on censoring-model slopes prevents separation and is negligible per observation in the fixed-dimensional regime. Only a vanishing numerical probability guard, $[N^{-2},1-N^{-2}]$ with $N=200J$, is used.

Let $h=A/g_t-(1-A)/(1-g_t)$. The first targeting step fits the observed-outcome fluctuation using covariate $h/\hat g_\Delta$ and predicts
\begin{equation}
 Q_Y^*=\hat Q_Y+\hat\epsilon\,h/\hat g_\Delta
 \label{eq:correct_target}
\end{equation}
on every record. The missingness indicator enters the fluctuation loss through the observed outcomes. Separate armwise fluctuations of $\hat Q_{\mathrm{int}}$ then solve the remaining score equations, and $\hat\Psi$ averages their predicted contrast. The reported plug-in variance uses cluster averages of
\begin{equation}
 d_{\hat\eta}=h\frac{\Delta}{\hat g_\Delta}(Y-Q_Y^*)
 +h\{Q_Y^*-Q_{\mathrm{int}}^*(A,W,t)\}
 +Q_{\mathrm{int}}^*(1,W,t)-Q_{\mathrm{int}}^*(0,W,t)-\hat\Psi.
 \label{eq:sim_eif}
\end{equation}
The targeting-score correspondence is established in Supplement S4.3 and S9.1.

We compare three nuisance configurations: the fixed-dimensional models; the same models with $p=\operatorname{round}(\kappa N_{\mathrm{train}})$ independent Gaussian noise covariates in both $Q_Y$ and $g_\Delta$, for $\kappa=0.05,0.10$; and outcome and censoring adjustments estimated from the evaluation cluster itself. The third configuration uses evaluation-cluster information in nuisance adjustment. Outcome-only and variance-shrunken adjustments are also examined. Noise covariates, their dimension, and fold labels are held fixed during jackknife deletion.

For every replicate, the true marginal nuisance functions give $D_j^0$, the cluster average of~\eqref{eq:sim_eif} at the truth. We compute $\bar D_J=J^{-1}\sum_jD_j^0$ and the total residual $R=\hat\Psi-\Psi_0-\bar D_J$. Table~\ref{tab:decomp} reports $V=J\widehat{\mathrm{Var}}(\hat\Psi)$, $V_D=J\widehat{\mathrm{Var}}(\bar D_J)$, $V_R=J\widehat{\mathrm{Var}}(R)$, and $C=2J\widehat{\mathrm{Cov}}(\bar D_J,R)$; their sample identity $V=V_D+V_R+C$ is checked before table generation. This $R$ is the total residual in the variance decomposition, including any empirical-process and score-solving components.

The main grid crosses $J=10,30,100$, ICC $0,0.013,0.05$, and the three fixed/growing nuisance dimensions. Replicate counts range from 60 to 300, with separate prespecified jackknife subsets. All 40 reported configurations and their counts are listed in the replication manifest; there are 7,560 point-estimator replicates and 2,860 full jackknife evaluations. Variance ratios use mean estimated variance divided by the empirical variance of the point estimator. Coverage uses $t_{J-1}$ critical values. Supplement S9 reports paired Monte Carlo intervals and exact binomial intervals for coverage.

\subsection{Simulation results}\label{sec:sim:growing}
At ICC zero, the fixed-dimensional residual-to-reference variance ratios are $0.24$ at $J=10$ and $0.02$ at $J=100$. The corresponding plug-in variance ratios are $0.88$ and $1.08$. For $\kappa=0.10$, the residual-to-reference ratios are $0.44$ and $0.20$, and the plug-in ratios are $42.55$ and $83.81$. Both nuisance specifications use the same estimation and simulation design.

In the $\kappa=0.10$, ICC-zero cell at $J=100$, the scaled decomposition is $V=25.1$, $V_D=23.1$, $V_R=4.6$ and $C=-2.7$, all multiplied by $10^3$. The cross term must be included when comparing the oracle variance with total variance. In the growing-dimensional configurations, the total residual and fitted score contribute separately to the variance discrepancy. The fixed-partition jackknife ratios in the ICC-zero, $\kappa=0.10$ configuration are $6.15$ at $J=10$ and $31.01$ at $J=30$. These finite-sample values indicate substantial jackknife variance overestimation.

\begin{table}[p]
\caption{Variance decomposition in the stepped-wedge simulation. $V,V_D,V_R,C$ are multiplied by $10^3$; $p$ is the number of added noise covariates.}\label{tab:decomp}
\centering\normalsize\setlength{\tabcolsep}{5pt}\renewcommand{\arraystretch}{1.18}
\begin{tabular}{rrrrrrrrr}
\toprule
ICC & $\kappa$ & $J$ & $p$ & $V$ & $V_D$ & $V_R$ & $C$ & $V_R/V_D$ \\
\midrule
0 & 0 & 10 & 0 & 29.6 & 21.9 & 5.2 & 2.4 & 0.24 \\
0 & 0 & 30 & 0 & 21.6 & 20.9 & 1.2 & -0.6 & 0.06 \\
0 & 0 & 100 & 0 & 22.0 & 22.3 & 0.4 & -0.6 & 0.02 \\
0 & 0.05 & 10 & 80 & 31.5 & 21.9 & 7.6 & 2.0 & 0.35 \\
0 & 0.05 & 30 & 240 & 24.8 & 20.9 & 4.8 & -0.8 & 0.23 \\
0 & 0.05 & 100 & 800 & 22.5 & 22.3 & 4.2 & -4.0 & 0.19 \\
0 & 0.10 & 10 & 160 & 31.0 & 21.9 & 9.6 & -0.6 & 0.44 \\
0 & 0.10 & 30 & 480 & 20.9 & 19.5 & 9.5 & -8.1 & 0.49 \\
0 & 0.10 & 100 & 1600 & 25.1 & 23.1 & 4.6 & -2.7 & 0.20 \\
0.013 & 0 & 10 & 0 & 71.7 & 61.1 & 13.2 & -2.6 & 0.22 \\
0.013 & 0 & 30 & 0 & 63.8 & 60.2 & 4.2 & -0.6 & 0.07 \\
0.013 & 0 & 100 & 0 & 67.2 & 65.3 & 1.0 & 0.9 & 0.02 \\
0.013 & 0.05 & 10 & 80 & 75.5 & 61.1 & 15.5 & -1.1 & 0.25 \\
0.013 & 0.05 & 30 & 240 & 66.9 & 60.2 & 8.6 & -1.8 & 0.14 \\
0.013 & 0.05 & 100 & 800 & 76.7 & 66.0 & 6.0 & 4.7 & 0.09 \\
0.013 & 0.10 & 10 & 160 & 71.6 & 61.1 & 16.2 & -5.7 & 0.26 \\
0.013 & 0.10 & 30 & 480 & 64.0 & 60.3 & 8.5 & -4.9 & 0.14 \\
0.013 & 0.10 & 100 & 1600 & 77.4 & 66.0 & 8.4 & 3.1 & 0.13 \\
0.05 & 0 & 10 & 0 & 215.2 & 181.2 & 42.7 & -8.8 & 0.24 \\
0.05 & 0 & 30 & 0 & 195.0 & 183.2 & 13.1 & -1.2 & 0.07 \\
0.05 & 0 & 100 & 0 & 181.6 & 175.6 & 3.0 & 3.0 & 0.02 \\
0.05 & 0.05 & 10 & 80 & 219.0 & 181.2 & 44.2 & -6.5 & 0.24 \\
0.05 & 0.05 & 30 & 240 & 197.0 & 183.2 & 16.6 & -2.8 & 0.09 \\
0.05 & 0.05 & 100 & 800 & 199.5 & 177.4 & 6.7 & 15.4 & 0.04 \\
0.05 & 0.10 & 10 & 160 & 206.0 & 181.2 & 38.5 & -13.7 & 0.21 \\
0.05 & 0.10 & 30 & 480 & 185.3 & 183.2 & 15.3 & -13.2 & 0.08 \\
0.05 & 0.10 & 100 & 1600 & 198.9 & 177.4 & 10.7 & 10.9 & 0.06 \\
\bottomrule
\end{tabular}

\end{table}
\begin{table}[p]
\caption{Variance ratios and 95\% interval coverage (CP). An empty entry means the jackknife was not run for that cell. Monte Carlo intervals are in Supplement S9.}\label{tab:varest}
\centering\normalsize\setlength{\tabcolsep}{5pt}\renewcommand{\arraystretch}{1.18}
\begin{tabular}{rrr|rr|rr|rr}
\toprule
 & & & \multicolumn{2}{c}{Plug-in} & \multicolumn{2}{c}{Oracle} & \multicolumn{2}{c}{Jackknife} \\
ICC & $\kappa$ & $J$ & Ratio & CP & Ratio & CP & Ratio & CP \\
\midrule
0 & 0 & 10 & 0.88 & 0.903 & 0.78 & 0.913 & 1.20 & 0.937 \\
0 & 0 & 30 & 1.08 & 0.950 & 1.04 & 0.955 & 1.13 & 1.000 \\
0 & 0 & 100 & 1.08 & 0.960 & 1.06 & 0.960 & -- & -- \\
0 & 0.05 & 10 & 1.04 & 0.907 & 0.73 & 0.897 & 1.43 & 0.967 \\
0 & 0.05 & 30 & 1.34 & 0.950 & 0.91 & 0.925 & 1.44 & 0.990 \\
0 & 0.05 & 100 & 1.64 & 0.990 & 1.04 & 0.960 & -- & -- \\
0 & 0.10 & 10 & 42.55 & 0.990 & 0.75 & 0.900 & 6.15 & 0.970 \\
0 & 0.10 & 30 & 82.25 & 1.000 & 1.06 & 0.933 & 31.01 & 1.000 \\
0 & 0.10 & 100 & 83.81 & 1.000 & 0.94 & 0.933 & -- & -- \\
0.013 & 0 & 10 & 0.94 & 0.920 & 0.85 & 0.943 & -- & -- \\
0.013 & 0 & 30 & 1.08 & 0.960 & 1.04 & 0.960 & -- & -- \\
0.013 & 0 & 100 & 1.01 & 0.940 & 1.01 & 0.940 & -- & -- \\
0.013 & 0.05 & 10 & 1.02 & 0.933 & 0.81 & 0.940 & -- & -- \\
0.013 & 0.05 & 30 & 1.24 & 0.960 & 0.99 & 0.955 & 1.28 & 0.970 \\
0.013 & 0.05 & 100 & 1.11 & 0.967 & 0.90 & 0.950 & -- & -- \\
0.013 & 0.10 & 10 & 7.97 & 0.987 & 0.85 & 0.943 & -- & -- \\
0.013 & 0.10 & 30 & 21.04 & 1.000 & 1.06 & 0.947 & -- & -- \\
0.013 & 0.10 & 100 & 25.32 & 1.000 & 0.90 & 0.933 & -- & -- \\
0.05 & 0 & 10 & 0.92 & 0.907 & 0.85 & 0.937 & -- & -- \\
0.05 & 0 & 30 & 1.05 & 0.935 & 1.01 & 0.950 & -- & -- \\
0.05 & 0 & 100 & 1.12 & 0.960 & 1.11 & 0.960 & -- & -- \\
0.05 & 0.05 & 10 & 0.99 & 0.917 & 0.83 & 0.940 & -- & -- \\
0.05 & 0.05 & 30 & 1.16 & 0.945 & 1.00 & 0.945 & -- & -- \\
0.05 & 0.05 & 100 & 1.17 & 0.983 & 1.03 & 0.950 & -- & -- \\
0.05 & 0.10 & 10 & 5.91 & 0.957 & 0.89 & 0.947 & -- & -- \\
0.05 & 0.10 & 30 & 297.35 & 1.000 & 1.06 & 0.960 & -- & -- \\
0.05 & 0.10 & 100 & 10.91 & 1.000 & 1.04 & 0.950 & -- & -- \\
\bottomrule
\end{tabular}

\end{table}

\subsection{Own-cluster adjustment results}\label{sec:sim:own}
With joint own-cluster outcome and censoring adjustments at ICC $0.05$, the plug-in variance ratios $0.18$, $0.23$, $0.21$ correspond to $J=10$, $30$, $100$; their coverage probabilities are $0.633$, $0.645$, $0.590$. The fixed-partition jackknife ratios are $1.23$, $1.15$, $1.05$. The outcome-only and shrunken comparisons, with their Monte Carlo uncertainty, appear in Supplement S9.

\begin{table}[htbp]
\caption{Own-cluster outcome and censoring adjustments at ICC $0.05$. The first panel gives the decomposition multiplied by $10^3$; the second gives variance ratios and coverage.}\label{tab:own}
\centering\small\renewcommand{\arraystretch}{1.18}
\begin{tabular}{rrrrr}
\toprule
$J$ & $V$ & $V_D$ & $V_R$ & $C$ \\
\midrule
10 & 123.7 & 181.2 & 37.4 & -95.0 \\
30 & 110.8 & 183.2 & 23.7 & -96.0 \\
100 & 119.0 & 175.6 & 15.8 & -72.4 \\
\bottomrule
\end{tabular}
\par\medskip
\begin{tabular}{r|rr|rr|rr}
\toprule
 & \multicolumn{2}{c}{Plug-in} & \multicolumn{2}{c}{Oracle} & \multicolumn{2}{c}{Jackknife} \\
$J$ & Ratio & CP & Ratio & CP & Ratio & CP \\
\midrule
10 & 0.18 & 0.633 & 1.48 & 0.977 & 1.23 & 0.953 \\
30 & 0.23 & 0.645 & 1.78 & 0.990 & 1.15 & 0.970 \\
100 & 0.21 & 0.590 & 1.69 & 0.990 & 1.05 & 0.967 \\
\bottomrule
\end{tabular}

\end{table}
The stylized model~(M) in Theorem~\ref{thm:own} provides an analytic comparison alongside these simulations.

\section{Discussion}
\label{sec:disc}

We study a simple but practically important point: for estimators with a von~Mises expansion, near-unbiasedness can coexist with an $O(n^{-1})$ contribution to the sampling variance. The finite-sample decomposition in Section~\ref{sec:decomp} makes that discrepancy explicit. When the remainder variance is non-negligible and the cross term is asymptotically negligible, a consistent influence-function sandwich targets only $\sigIF/n$, whereas the total variance is $(\sigIF+c_R)/n$ to first order.

Jackknife consistency depends on a leave-one-out remainder linearization, while bootstrap variance consistency requires conditional Mallows-2 convergence. The corresponding assumptions describe distinct forms of stability and distributional approximation; their verification depends on the nuisance fitting procedure and data structure.

Persistent disagreement between sandwich and resampling variance estimates may reflect a structural discrepancy rather than only a small-sample artifact. In the own-cluster configurations studied, the fixed-partition jackknife improves variance estimation; in the growing-dimensional configurations, however, it can substantially overestimate variance. The distinction between remainder variation and first-order score mismatch is consequential for interpreting such variance ratios.

Several questions remain open. It would be useful to develop direct, easily verifiable conditions for the near-boundary central limit theorem, sharper non-asymptotic error bounds for jackknife-based pivots, and analytic estimators of the remainder-variance contribution that avoid full resampling. Those extensions would further clarify when the near-boundary regime is merely a cautionary case and when it is a routine feature of modern semiparametric practice.

\section*{Reproducibility}
The data-generating mechanisms, estimators, simulation settings, replication counts,
and uncertainty calculations are described in Supplement S9. The supplementary tables
provide numerical results for all reported configurations.

\bibliographystyle{plainnat}
\bibliography{references}

\clearpage
\appendix
\setcounter{section}{0}
\setcounter{table}{0}
\setcounter{equation}{0}
\renewcommand{\thesection}{S\arabic{section}}
\renewcommand{\thetable}{S\arabic{table}}
\renewcommand{\theequation}{S\arabic{equation}}
\renewcommand{\theHequation}{supp.\arabic{equation}}
\renewcommand{\theHtable}{supp.\arabic{table}}
\renewcommand{\theHsection}{supp.\arabic{section}}

\section*{Supplementary Material}

Section, theorem, proposition, and assumption numbers beginning with
``Section~$n$'', ``Theorem~$n.k$'', etc.\ refer to the main paper.

\section{Leave-One-Out Remainder Linearization}
\label{sec:sm:loo}

The LOO Remainder Linearization assumption (Assumption~5.1 of the main paper)
requires a uniform linearization
$\max_i|nb_i + \psi_{n,i}|\xrightarrow{P}0$ together with the four moment
conditions (i)--(iv) stated there.  Controlling the leave-one-out perturbations
$b_i = \Rrem^{(-i)} - \Rrem$ \emph{uniformly over all $n$ indices} is the key
analytical task.  This is structurally different from standard empirical
process results, which control $\Rrem(P_n)$ as a functional of the full
empirical measure; the LOO perturbation involves removing one observation
from the nuisance training data, and uniform control over $i$ requires a
separate one-point-contamination argument.

\subsection*{S1.1\; Regular plug-in learners}

\begin{sproposition}[Smooth plug-in sufficient condition]
\label{sprop:hadamard}
Suppose:
\begin{enumerate}[label=\emph{(\alph*)}]
  \item \emph{(Differentiability.)}
    $P \mapsto \Rrem(P)$ is Hadamard differentiable at $P_0$
    with derivative $\dot\Rrem_{P_0}[\nu] = \int \psi \, d\nu$
    for some $\psi \in L^2_0(P_0)$ with $\mathrm{Var}(\psi)=c_R$
    and $\EE[\psi^4]<\infty$.
  \item \emph{(Cross-orthogonality.)}
    $\EE\bigl[D^*(O;\eta^0)\,\psi(O)\bigr]=0$.
  \item \emph{(Uniform LOO expansion.)}
    $\sup_{1\le i\le n}|n(\Rrem^{(-i)}-\Rrem)+\psi(O_i)|\xrightarrow{P}0$.
\end{enumerate}
Then the LOO Remainder Linearization assumption (Assumption~5.1) holds with
$\psi_{n,i}=\psi(O_i)$: conditions (i)--(iv) of Assumption~5.1 follow from
\begin{itemize}[noitemsep,topsep=0pt]
  \item the weak law of large numbers, giving
        $n^{-1}\sum_i\psi(O_i)=o_P(1)$ [(i)];
  \item the strong law applied to $\psi(O_i)^2$, giving
        $n^{-1}\sum_i\psi(O_i)^2\xrightarrow{P}c_R$ [(ii)];
  \item (b) and the WLLN imply
        $n^{-1}\sum_i D^*(O_i;\eta^0)\psi(O_i)\xrightarrow{P}0$ [(iii)];
  \item the bound
        $n^{-1}\EE[\max_i\psi(O_i)^2]\le n^{-1}\sum_i\EE[\psi(O_i)^2]
         =c_R+o(1)<\infty$
        together with $\sup_i\EE[\psi(O_i)^4]<\infty$ by (a) [(iv)].
\end{itemize}
\end{sproposition}

The uniform deletion expansion in condition~(c) is stronger than
Hadamard differentiability alone and requires control of the
one-observation perturbation. For a correctly specified smooth
fixed-dimensional bilinear remainder, the first derivative at $P_0$ is
zero, giving $\psi=0$ and $c_R=0$. Growing-dimensional models can have
a different limit, as shown in Supplement S10; see also
\citet[Chapter~20]{van_der_vaart1998asymptotic}.

\subsection*{S1.2\; Data-adaptive learners}

For a fixed-weight bilinear remainder with
$\max_k\|\hat\eta_k-\eta_k^0\|_{P_0}=O_P(r_n)$ and $r_n\to0$, suppose
\[\max_{1\le i\le n}\sup_x|\hat\eta_k(\mathcal D_n;x)-\hat\eta_k(\mathcal D_n^{(-i)};x)|=O(n^{-1})\quad\text{a.s.}\]
The resulting uniform bound is
$\max_i|nb_i|=O_P(r_n+n^{-1})\xrightarrow{P}0$.
For adaptive learners, leave-one-out behavior and the quadratic limit
in Assumption~5.1 depend on the learner and training construction.

\section{Bootstrap Variance Consistency}
\label{sec:sm:boot}

\begin{stheorem}[Bootstrap Variance Consistency]
Under Assumptions~2.1--2.4 and Assumption~5.2 (Mallows-2) from the main
paper, in the iterated limit $B\to\infty$ followed by $n\to\infty$, the pairs bootstrap variance satisfies
\[
  n\{\VarBoot(\Psihat)-\mathrm{Var}(\Psihat-\Psitrue)\}\xrightarrow{P}0.
\]
Equivalently,
\[
  n\VarBoot(\Psihat) - \mathrm{Var}(T_n) \xrightarrow{P} 0,
  \qquad T_n=\sqrt{n}(\Psihat-\Psitrue).
\]
Thus the bootstrap targets the full sampling variance, including any
non-negligible remainder variance and covariance terms appearing in the
variance decomposition of Theorem~3.1 of the main paper.
\end{stheorem}

\begin{proof}
\textit{Step 1: Oracle bootstrap variance.}
Let $T_n = \sqrt{n}(\Psihat-\Psitrue)$ and $T_n^* = \sqrt{n}(\hat\Psi^*-\Psihat)$
be the centred bootstrap statistic.  Define the \emph{oracle} bootstrap variance
$V_n^{\rm oracle} := \mathrm{Var}^*[T_n^*\mid\mathcal{D}_n]$, the exact
conditional variance of $T_n^*$ given the observed data.

By the Mallows-2 bootstrap assumption, $d_2(T_n^*,T_n)\xrightarrow{P}0$, where
$d_2$ is the Wasserstein-2 distance.  Since Wasserstein-2 convergence is
equivalent to weak convergence plus convergence of second moments
\citep{bickel1981some}, and since $\EE^*[T_n^*\mid\mathcal{D}_n]\xrightarrow{P}
\EE[T_n]$ under the same Mallows-2 assumption,
\begin{equation}
  V_n^{\rm oracle}
  = \mathrm{Var}^*[T_n^*\mid\mathcal{D}_n]
  = \mathrm{Var}(T_n)+o_P(1).
  \label{eq:oracle_conv}
\end{equation}
The target $\mathrm{Var}(T_n)$ is exactly the variance of the full centred
estimator.  By Theorem~3.1 of the main paper,
$\mathrm{Var}(T_n)=\sigIF + n\mathrm{Var}(\Rrem)+2n\mathrm{Cov}(\bar D_n,\Rrem)$,
so the bootstrap target automatically includes any non-negligible remainder
variance and covariance contributions.

\textit{Step 2: Oracle equals Monte Carlo in the $B\to\infty$ limit.}
The bootstrap variance estimator $\VarBoot(\Psihat)$ is defined as
the sample variance of $B$ conditionally i.i.d.\ bootstrap replicates
$\hat\Psi^{*1},\ldots,\hat\Psi^{*B}$:
\[
  \VarBoot(\Psihat)
  = \frac{1}{B-1}\sum_{b=1}^B
    \bigl(\hat\Psi^{*b}-\bar\Psi^*\bigr)^2,
  \quad
  \bar\Psi^* = \frac{1}{B}\sum_{b=1}^B\hat\Psi^{*b}.
\]
By the conditional law of large numbers applied to the bootstrap sample
variance around the bootstrap sample mean, as $B\to\infty$,
\[
  n\VarBoot(\Psihat)
  \;\xrightarrow{a.s.\mid\mathcal{D}_n}\;
  \mathrm{Var}^*(T_n^*\mid\mathcal{D}_n)
  = V_n^{\rm oracle}.
\]
Combining with~\eqref{eq:oracle_conv}:
\[
  n\VarBoot(\Psihat) \;\xrightarrow{P}\; \mathrm{Var}(T_n),
\]
equivalently
\[
  n\{\VarBoot(\Psihat)-\mathrm{Var}(\Psihat-\Psitrue)\}\xrightarrow{P}0.
\]
These statements use the iterated limit. For a joint sequence $B_n\to\infty$,
assume additionally $\EE^*[|T_n^*|^4\mid\mathcal D_n]=O_P(1)$.
The conditional variance of the Monte Carlo second moment is at most
$B_n^{-1}\EE^*[|T_n^*|^4\mid\mathcal D_n]$; conditional Chebyshev and the
same argument for the mean give an $o_P(1)$ Monte Carlo error.
A bounded unconditional second moment alone does not justify exchanging limits.
\end{proof}

\section{Jackknife Wald Interval Validity}
\label{sec:sm:jk_ci}

\begin{stheorem}[Jackknife Wald Interval Validity]
Under the conditions of Theorem~5.1 (Jackknife Variance Consistency),
Assumption~5.3 (Near-Boundary CLT), and the negligible-cross-term condition,
the studentized jackknife pivot
$T_n^{\rm JK} = (\Psihat - \Psitrue)/\sqrt{\VarJK} \Rightarrow N(0,1)$,
and the jackknife Wald interval is asymptotically valid.
\end{stheorem}

\begin{proof}
By Theorem~5.1 and the negligible-cross-term condition,
$n\VarJK\to_P\sigIF+c_R>0$. Assumption~5.3 gives
$\sqrt n(\Psihat-\Psitrue)\Rightarrow N(0,\sigIF+c_R)$.
Consequently,
\[
 T_n^{\rm JK}=\frac{\sqrt n(\Psihat-\Psitrue)}{\sqrt{n\VarJK}}
 \Rightarrow N(0,1)
\]
by Slutsky's theorem. The argument is identical in the strong-decay case
$c_R=0$. This establishes the studentized limit.
\end{proof}

\section{Secondary Simulation: Strong-Decay Regime (AIPW)}
\label{sec:sm:aipw}

This section illustrates the strong-decay regime using a non-cross-fitted
augmented inverse probability weighted (AIPW) estimator in an i.i.d.\
observational study.  The purpose is to illustrate how the variance ratio
$\hat\rho = \VarJK/\VarSand$ decreases to $1.0$ as $n$ grows when nuisance
estimation operates in the strong-decay regime
($n\cdot\mathrm{Var}(\Rrem)\to 0$). This example uses an i.i.d. AIPW estimator with parametric nuisance fits, in contrast to the stepped-wedge estimator examined in Section~S9.

\subsection*{S4.1\; Data-generating process}

\begin{alignat*}{3}
  W &\sim N(0,1), \quad
  &A \mid W &\sim \mathrm{Bernoulli}(\expit(0.3W)),\\
  Y &= 0.5 + 0.4A + 0.3W + 0.15AW + \varepsilon, \quad
  &\varepsilon &\sim N(0, 0.5^2).
\end{alignat*}
True ATE: $\Psi_0 = 0.4$.  Nuisance models: outcome regression correctly
specified as $\mathrm{lm}(Y \sim A \times W)$; propensity enriched as
$\mathrm{glm}(A \sim W + W^2)$.  No cross-fitting ($V=1$) is used so that
the LOO and full-data fits are structurally comparable.  Under correctly
specified parametric nuisance with $L^4$ errors of order $n^{-1/2}$ and suitable inverse-propensity and residual moment control,
$n\cdot\mathrm{Var}(\Rrem)\to 0$, placing the estimator in the strong-decay
regime (Theorem~4.1 of the main paper).  By Cauchy--Schwarz,
$|n\,\mathrm{Cov}(\bar D_n,\Rrem)|
\le \{n\,\mathrm{Var}(\bar D_n)\}^{1/2}\{n\,\mathrm{Var}(\Rrem)\}^{1/2}
= O(1)\cdot o(1) = o(1)$,
so the cross-term condition of Remark~3.2 also holds here even without
cross-fitting.

\subsection*{S4.2\; Results}

\begin{table}[h]
\centering
\caption{AIPW secondary simulation: strong-decay regime.
  $\hat\rho = \bar{V}_{\rm JK}/\bar{V}_{\rm Sand}$ decreases monotonically
  to $1.0$, consistent with the strong-decay interpretation under the stated
  moment conditions. CP(Sand) and CP(JK): empirical 95\% CI coverage ($z_{\alpha/2}$ critical values).
  $R = 500$ replicates; MC-SE $\approx 0.010$.}
\label{tab:aipw}
\smallskip
\begin{tabular}{crrrrr}
\toprule
$n$ & Bias & MCSD & CP(Sand) & CP(JK) & $\hat\rho$ \\
\midrule
 200 & $+0.000$ & $0.070$ & $0.964$ & $0.968$ & $1.040$ \\
 500 & $-0.002$ & $0.045$ & $0.946$ & $0.946$ & $1.014$ \\
1000 & $-0.001$ & $0.031$ & $0.960$ & $0.960$ & $1.007$ \\
2000 & $-0.000$ & $0.023$ & $0.958$ & $0.958$ & $1.003$ \\
\bottomrule
\end{tabular}
\end{table}

$\hat\rho$ decreases monotonically from $1.040$ to $1.003$, consistent
with the sharp sandwich characterization (Theorem~4.1 of the main paper):
under the stated nuisance moment and deletion conditions, $n\VarSand \xrightarrow{P} \sigIF$
and $n\VarJK \xrightarrow{P} \sigIF$, so no correction beyond the sandwich
is needed.  All coverage probabilities lie within $\pm 2$ Monte Carlo
standard errors of $0.95$.  These results illustrate the strong-decay setting separately from the clustered simulations.

\subsection*{S4.3\; Targeting and influence-function correspondence}
Let $A\sim\mathrm{Bernoulli}(1/2)$, $\Delta\sim\mathrm{Bernoulli}(r)$, and $Y\sim N(0,1)$ be independent, with known nuisance functions $Q_Y=Q_{\mathrm{int}}=0$ and $g_\Delta=r$. Put $h=4A-2$. The efficient influence function for the null mean contrast is $D=h\Delta Y/r$, with $\mathrm{Var}(D)=4/r$.

The observed-outcome least-squares fluctuation coefficient is
\[
 \hat\epsilon=
 \frac{\sum_i(h_i/r)\Delta_iY_i}{\sum_i\Delta_i(h_i/r)^2}
 =\frac r4 P_nD+o_P(n^{-1/2}).
\]
The targeting step sets $Q_Y^*=\hat\epsilon h/r$ across all records.
After the armwise integration update, the estimator satisfies
$\hat\Psi=(4/r)\hat\epsilon=P_nD+o_P(n^{-1/2})$.
A Monte Carlo check at $r=1/2$ uses 2,000 replicates for each
of $N=500$ and $N=2000$. Table~\ref{tab:arxiv_s2} compares the
estimator's empirical variance with that of the oracle influence-function
average and reports the variance of their difference.

\begin{table}[p]
\centering
\small
\caption{Known-nuisance check of the targeting estimator, with 2,000 replicates per row. The slope is from regression of the estimator on the oracle influence-function average. Variances are multiplied by $N$.}
\label{tab:arxiv_s2}
\setlength{\tabcolsep}{6pt}
\begin{tabular}{rrrrr}
\toprule
$N$ & Slope & Estimator variance & Oracle variance & Residual variance \\
\midrule
500  & 0.9993 & 8.244 & 8.238 & 0.0161 \\
2000 & 0.9998 & 7.802 & 7.801 & 0.0036 \\
\bottomrule
\end{tabular}
\end{table}

\section{Additional Proofs for Core Main-Text Results}
\label{sec:sm:core_proofs}

This section records short algebraic proofs for several results that are only
sketched in the main paper.  We write
\[
  \bar D_n := \frac{1}{n}\sum_{i=1}^n D^*(O_i;\eta^0),
  \qquad
  R_n := \Rrem,
\]
so that the von Mises expansion is
\[
  \Psihat - \Psitrue = \bar D_n + R_n.
\]
Throughout this section, $\mathbb E[D^*(O;\eta^0)] = 0$ and
$\mathrm{Var}\{D^*(O;\eta^0)\}=\sigIF$.

\subsection*{S5.1\; Finite-sample variance decomposition (Theorem~3.1)}

\begin{proof}[Proof of Theorem~3.1 of the main paper]
By the von Mises expansion,
\[
  \mathrm{Var}(\Psihat-\Psitrue)
  = \mathrm{Var}(\bar D_n + R_n)
  = \mathrm{Var}(\bar D_n) + 2\,\mathrm{Cov}(\bar D_n,R_n) + \mathrm{Var}(R_n).
\]
Since the $D^*(O_i;\eta^0)$ are i.i.d., mean zero, and have variance
$\sigIF$,
\[
  \mathrm{Var}(\bar D_n)
  = \frac{1}{n^2}\sum_{i=1}^n \mathrm{Var}\{D^*(O_i;\eta^0)\}
  = \frac{\sigIF}{n}.
\]
Substituting yields
\[
  \mathrm{Var}(\Psihat-\Psitrue)
  = \frac{\sigIF}{n} + 2\,\mathrm{Cov}(\bar D_n,R_n) + \mathrm{Var}(R_n),
\]
which is the claimed decomposition.
\end{proof}

\subsection*{S5.2\; Boundary variance gap and sandwich characterization}

\begin{proof}[Proof of Proposition~3.1 (Variance Gap at the Product-Rate Boundary)]
Assume the standard bilinear remainder structure together with the extra
conditions stated in the proposition:
\[
  n\,\mathrm{Cov}(\bar D_n,R_n) \to 0,
  \qquad
  c_R := \lim_{n\to\infty} n\,\mathrm{Var}(R_n) \in (0,\infty).
\]
Then the decomposition above gives
\[
  \mathrm{Var}(\Psihat-\Psitrue)
  = \frac{\sigIF}{n} + o(n^{-1}) + \frac{c_R}{n} + o(n^{-1})
  = \frac{\sigIF+c_R}{n} + o(n^{-1}).
\]
Hence
\[
  \mathrm{Var}(\Psihat-\Psitrue) - \frac{\sigIF}{n}
  = \frac{c_R}{n} + o(n^{-1})
  = \Theta(n^{-1}),
\]
and division by $\sigIF/n$ yields
\[
  \frac{\mathrm{Var}(\Psihat-\Psitrue)}{\sigIF/n}
  = 1 + \frac{c_R}{\sigIF} + o(1).
\]
Therefore the relative gap does not vanish unless $c_R=0$.
\end{proof}

\begin{proof}[Proof of Theorem~4.1 and Corollary~4.1]
The influence-function sandwich estimator is built from the empirical second
moment of the leading linear term. The oracle law of large numbers, or the additional empirical $L^2$ plug-in condition stated in Section~4, gives
\[
  n\,\VarSand(\Psihat) \xrightarrow{P} \sigIF.
\]
Now consider the two regimes.

\smallskip
\noindent\emph{Strong decay.}
If $n\,\mathrm{Var}(R_n)\to 0$ and
$n\,\mathrm{Cov}(\bar D_n,R_n)\to 0$, then the decomposition in S5.1 gives
\[
  n\,\mathrm{Var}(\Psihat-\Psitrue)
  = \sigIF + o(1),
\]
and the sandwich estimator is asymptotically correct for the full sampling
variance.

\smallskip
\noindent\emph{Near boundary.}
If $n\,\mathrm{Var}(R_n)\to c_R>0$ and
$n\,\mathrm{Cov}(\bar D_n,R_n)\to 0$, then
\[
  n\,\mathrm{Var}(\Psihat-\Psitrue)
  = \sigIF + c_R + o(1),
\]
whereas still $n\,\VarSand(\Psihat) \xrightarrow{P} \sigIF$.  Therefore
\[
  \mathrm{Var}(\Psihat-\Psitrue) - \VarSand(\Psihat)
  = \frac{c_R}{n} + o_P(n^{-1}),
\]
so the sandwich underestimates the total variance by a non-vanishing
first-order amount.

For the corollary, $|R_n|\le C\|\hat\eta_1-\eta_1^0\|
\|\hat\eta_2-\eta_2^0\|+|s_n|$ and $(a+b)^2\le2a^2+2b^2$
give $n\EE R_n^2\to0$ under the displayed moment conditions.
Alternatively, convergence in probability of $nR_n^2$ to zero plus uniform
integrability implies convergence of its expectation.
Faster probability rates alone are insufficient: let
$B_n=\mathbf1\{U_1\le1/n\}$ for an independent uniform component of the data,
and take both nuisance errors equal to the constant function $B_n$.
They are $O_P(n^{-a})$ for every fixed $a>0$, but $R_n=B_n^2=B_n$
has $n\mathrm{Var}(R_n)=1-1/n\to1$.
Here $\sqrt nR_n\xrightarrow{P}0$, so a usual centered CLT and Wald coverage can
remain valid despite failure of actual second-moment convergence.
\end{proof}

\subsection*{S5.3\; Full proof of the LOO expansion lemma (Lemma~5.1)}

\begin{proof}[Proof of Lemma~5.1 of the main paper]
Write $D_i:=D^*(O_i;\eta^0)$ and
\[
  \bar D_n^{(-i)} := \frac{1}{n-1}\sum_{j\neq i} D_j.
\]
The von Mises expansions for the full-sample and leave-one-out estimators are
\[
  \Psihat-\Psitrue = \bar D_n + R_n,
  \qquad
  \hat\Psi^{(-i)}-\Psitrue = \bar D_n^{(-i)} + R_n^{(-i)}.
\]
Subtracting,
\[
  \hat\Psi^{(-i)}-\Psihat = \bigl(\bar D_n^{(-i)}-\bar D_n\bigr) + \bigl(R_n^{(-i)}-R_n\bigr)
  = \bigl(\bar D_n^{(-i)}-\bar D_n\bigr) + b_i.
\]
Now
\[
  \bar D_n^{(-i)}
  = \frac{1}{n-1}\sum_{j\neq i} D_j
  = \frac{n\bar D_n-D_i}{n-1},
\]
so
\[
  \bar D_n^{(-i)}-\bar D_n
  = \frac{\bar D_n-D_i}{n-1}
  = -\frac{D_i-\bar D_n}{n-1}.
\]
Therefore
\[
  \hat\Psi^{(-i)}-\Psihat
  = -\frac{D_i-\bar D_n}{n-1} + b_i,
\]
which is the first displayed identity.

For the alternative representation, add and subtract $D_i/n$:
\[
  -\frac{D_i-\bar D_n}{n-1}
  = -\frac{D_i}{n} + \frac{n\bar D_n-D_i}{n(n-1)}
  = -\frac{D_i}{n} + a_i,
\]
where $a_i := (n\bar D_n-D_i)/\{n(n-1)\}$.
Hence
\[
  \hat\Psi^{(-i)}-\Psihat = -\frac{1}{n}D_i + b_i + a_i.
\]
It remains to verify the bounds on $a_i$. Under Assumptions~2.1--2.3,
$\bar D_n=O_P(n^{-1/2})$ and
$\max_i|D_i|\le(\sum_iD_i^2)^{1/2}=O_P(n^{1/2})$.
Consequently
\[
 \max_i|a_i|\le\frac{n|\bar D_n|+\max_i|D_i|}{n(n-1)}
 =O_P(n^{-3/2}),
 \qquad n^{-1}\sum_i a_i^2\le\max_i a_i^2=O_P(n^{-3}).
\]
These imply the (weaker) bounds stated in the lemma without requiring bounded outcomes.
This proves the lemma.
\end{proof}

\subsection*{S5.4\; Full proof of the stabilization proposition (Proposition~5.1)}

\begin{proof}[Proof of Proposition~5.1 of the main paper]
Let
\[
  r_{n,i} := nb_i + \psi_{n,i},
  \qquad
  \bar r_n := n^{-1}\sum_{i=1}^n r_{n,i},
  \qquad
  \bar\psi_n := n^{-1}\sum_{i=1}^n \psi_{n,i}.
\]
Assumption~5.1 gives $\max_i|r_{n,i}|\to_P 0$ and
$\bar\psi_n=o_P(1)$.  Since $|\bar r_n|\le \max_i|r_{n,i}|$, also
$\bar r_n=o_P(1)$.  Rearranging the definition of $r_{n,i}$,
\[
  n(b_i-\bar b)
  = -\bigl(\psi_{n,i}-\bar\psi_n\bigr) + \bigl(r_{n,i}-\bar r_n\bigr).
\]
Hence
\[
  C_n
  := (n-1)\sum_{i=1}^n (b_i-\bar b)^2
  = \frac{n-1}{n^2}\sum_{i=1}^n
    \Bigl\{\bigl(\psi_{n,i}-\bar\psi_n\bigr)-\bigl(r_{n,i}-\bar r_n\bigr)\Bigr\}^2.
\]
Define the comparison quadratic form
\[
  Q_n := \frac{n-1}{n^2}\sum_{i=1}^n (\psi_{n,i}-\bar\psi_n)^2.
\]
By the elementary inequality $|u^2-v^2|\le 2|v||u-v|+|u-v|^2$,
\[
  |C_n-Q_n|
  \le \frac{n-1}{n^2}\sum_{i=1}^n
    \Bigl[2|\psi_{n,i}-\bar\psi_n|\,|r_{n,i}-\bar r_n| + (r_{n,i}-\bar r_n)^2\Bigr].
\]
Since $|r_{n,i}-\bar r_n|\le 2\max_j|r_{n,j}|$,
\[
  |C_n-Q_n|
  \le \frac{4(n-1)}{n^2}\max_j|r_{n,j}|\sum_{i=1}^n |\psi_{n,i}-\bar\psi_n|
      + \frac{4(n-1)}{n}\max_j r_{n,j}^2.
\]
Also
\[
  \frac{1}{n}\sum_{i=1}^n |\psi_{n,i}-\bar\psi_n|
  \le \Bigl(\frac{1}{n}\sum_{i=1}^n (\psi_{n,i}-\bar\psi_n)^2\Bigr)^{1/2}
  \le \Bigl(\frac{1}{n}\sum_{i=1}^n \psi_{n,i}^2\Bigr)^{1/2} + |\bar\psi_n|
  = O_P(1)
\]
by Assumption~5.1(ii).  Therefore the first term is $o_P(1)$ because
$\max_j|r_{n,j}|=o_P(1)$, and the second term is also $o_P(1)$.  Hence
$C_n - Q_n \xrightarrow{P} 0$.

Next observe that
\[
  Q_n
  = \frac{n-1}{n}\left\{\frac{1}{n}\sum_{i=1}^n \psi_{n,i}^2 - \bar\psi_n^2\right\}.
\]
By Assumption~5.1(i)--(ii), $\bar\psi_n^2=o_P(1)$ and
$n^{-1}\sum_i \psi_{n,i}^2\to_P c_R$, so $Q_n \xrightarrow{P} c_R$.
Combining gives $C_n\to_P c_R$, as claimed.
\end{proof}

\subsection*{S5.5\; Full proof of the jackknife consistency theorem (Theorem~5.1)}

\begin{proof}[Proof of Theorem~5.1 of the main paper]
Write $D_i:=D^*(O_i;\eta^0)$, $x_i:=D_i-\bar D_n$, and $y_i:=b_i-\bar b$.
By Lemma~5.1,
\[
  \hat\Psi^{(-i)}-\Psihat = -\frac{x_i}{n-1} + b_i.
\]
Averaging over $i$ and using $\sum_i x_i=0$ gives
$\bar\Psi_{(-\cdot)}-\Psihat = \bar b$.  Therefore
\[
  \hat\Psi^{(-i)}-\bar\Psi_{(-\cdot)} = -\frac{x_i}{n-1} + y_i.
\]
Substituting into the jackknife definition,
\[
  n\VarJK(\Psihat)
  = (n-1)\sum_{i=1}^n \left(-\frac{x_i}{n-1}+y_i\right)^2
  = A_n + B_n + C_n,
\]
where
\[
  A_n := \frac{1}{n-1}\sum_{i=1}^n x_i^2,
  \qquad
  B_n := -2\sum_{i=1}^n x_i y_i,
  \qquad
  C_n := (n-1)\sum_{i=1}^n y_i^2.
\]

\smallskip
\noindent\emph{Step 1: the influence-function term.}
Since $x_i=D_i-\bar D_n$,
\[
  A_n = \frac{n}{n-1}\left(\frac{1}{n}\sum_{i=1}^n D_i^2-\bar D_n^2\right)
  \xrightarrow{P} \sigIF
\]
by the law of large numbers and $\mathbb E[D_i]=0$.

\smallskip
\noindent\emph{Step 2: the mixed term in the strong-decay regime.}
By Cauchy--Schwarz,
\[
  |B_n| \le 2\Bigl(\sum_{i=1}^n x_i^2\Bigr)^{1/2}
             \Bigl(\sum_{i=1}^n y_i^2\Bigr)^{1/2}.
\]
Now $\sum_i x_i^2 = O_P(n)$, while
$\sum_i y_i^2 \le \sum_i b_i^2 = o_P(n^{-1})$ under the strong-decay hypothesis.
Hence $B_n=o_P(1)$.  The same bound gives
$C_n \le (n-1)\sum_i b_i^2 = o_P(1)$.
Therefore, in the strong-decay regime,
$n\VarJK(\Psihat) = A_n + o_P(1) \xrightarrow{P} \sigIF$.

\smallskip
\noindent\emph{Step 3: the mixed term in the near-boundary regime.}
Under Assumption~5.1,
\[
  y_i = -\frac{\psi_{n,i}-\bar\psi_n}{n} + \frac{r_{n,i}-\bar r_n}{n},
\]
with $\max_i|r_{n,i}|\to_P 0$.  Therefore
\[
  B_n
  = \frac{2}{n}\sum_{i=1}^n x_i(\psi_{n,i}-\bar\psi_n)
    - \frac{2}{n}\sum_{i=1}^n x_i(r_{n,i}-\bar r_n).
\]
Because $\sum_i x_i=0$, the $\bar\psi_n$ and $\bar r_n$ terms vanish, so
\[
  B_n = \frac{2}{n}\sum_{i=1}^n x_i\psi_{n,i}
        - \frac{2}{n}\sum_{i=1}^n x_i r_{n,i}.
\]
For the first sum, using Assumption~5.1(i)--(iii),
\[
  \frac{1}{n}\sum_{i=1}^n x_i\psi_{n,i}
  = \frac{1}{n}\sum_{i=1}^n D_i\psi_{n,i} - \bar D_n\bar\psi_n
  = o_P(1).
\]
Note the essential role of Assumption~5.1(iii) (cross-orthogonality)
here: without it, the first term on the right would not vanish and the
jackknife would fail to isolate $c_R$.  For the second sum,
\[
  \left|\frac{1}{n}\sum_{i=1}^n x_i r_{n,i}\right|
  \le \max_i|r_{n,i}|\cdot \frac{1}{n}\sum_{i=1}^n |x_i|
  = o_P(1),
\]
because $n^{-1}\sum_i |x_i| = O_P(1)$ and $\max_i|r_{n,i}|=o_P(1)$.
Thus $B_n=o_P(1)$ in the near-boundary regime as well.

Finally, Proposition~5.1 (proved in S5.4) gives
$C_n \xrightarrow{P} c_R$.
Combining with Step~1 and $B_n=o_P(1)$:
\[
  n\VarJK(\Psihat) = A_n + B_n + C_n \xrightarrow{P} \sigIF + c_R.
\]
This proves both parts of the theorem.
\end{proof}

\section{Proof of Proposition~6.1 (Conditional Cluster Variance)}
\label{sec:sm:icc}
Conditional on an independent training sample, let $\Sigma_{\mathcal T}$ be the covariance matrix of the $M$ contributions in one evaluation cluster. Then
\[
 \mathrm{Var}(R\mid\mathcal T)=\frac{\mathbf1^\top\Sigma_{\mathcal T}\mathbf1}{JM^2}.
\]
Exchangeability reduces the numerator to $Mv_{\mathcal T}+M(M-1)c_{\mathcal T}$. Taking expectations and adding the variance of $\EE(R\mid\mathcal T)=\mu_{\mathcal T}$ proves Proposition~\ref{prop:icc}. The non-exchangeable expression is applicable to repeated periods with different treatment probabilities, provided training remains external. It does not apply directly after conditioning on all cross-fitting training samples, since their union includes the evaluation observations. Outcome ICC alone does not specify $\Sigma_{\mathcal T}$ for a learned score. \qed

\section{Proof of the Foldwise Plug-In Expansion}
\label{sec:sm:plugin}
Fix fold $k$ and condition only on its training sample $\mathcal T_{-k}$. The evaluation observations remain i.i.d. and the function $d_{n,k}$ is fixed. Conditional Chebyshev gives
\[
 \PP\left(\left|P_{n,k}d_{n,k}^2-P_0d_{n,k}^2\right|>\epsilon\mid\mathcal T_{-k}\right)
 \le\min\left(1,\frac{P_0d_{n,k}^4}{n_k\epsilon^2}\right).
\]
When $P_0d_{n,k}^4=O_P(1)$, first restrict to an event where it is bounded by $M$, then send $n_k\to\infty$ and $M\to\infty$. The unconditional error probability tends to zero. Summing over a fixed number of folds is valid without independence between folds. Thus
\[
 \sum_k\pi_k P_{n,k}d_{n,k}^2
 =\sum_k\pi_kP_0d_{n,k}^2+o_P(1)
 =\mathrm{Var}(D_n^*)+\Lambda_n+2\Gamma_n+o_P(1).
\]
The identical argument for first moments gives the global empirical mean $m_n+o_P(1)$. Since $m_n=O_P(1)$, centering subtracts $m_n^2+o_P(1)$. It does not subtract $\sum_k\pi_k(P_0d_{n,k})^2$: that would be separate centering within each fold.

For estimated functions $\tilde d_i$ that also use evaluation observations, a sufficient extension is
\[
 n^{-1}\sum_i\{\tilde d_i-d_{n,k(i)}(O_i)\}^2=o_P(1).
\]
Cauchy--Schwarz then transfers both first and second empirical moments, because the original second moment is $O_P(1)$. A finite-dimensional targeting fluctuation can satisfy this condition in a regular regime if its coefficients vanish and its covariates have bounded empirical moments. This condition must be established for the actual fluctuation; it is not supplied by the initial cross-fitting labels. The high-dimensional simulation diagnostic reports the empirical differences rather than assuming them negligible. \qed
\section{Proof of Theorem~6.1 (Own-Cluster Adjustment)}
\label{sec:sm:own}

Throughout, $m$, $T$ and $V$ are fixed and $J\to\infty$; all $O_P$ and $o_P$
statements are in $J$.  We use the notation of Section~6.3 of the main paper
and write $\bar\varepsilon_{jt} := m^{-1}\sum_i\varepsilon_{ijt}$,
$\bar\varepsilon_j := (mT)^{-1}\sum_{i,t}\varepsilon_{ijt}$,
$Q^0(a,t) := \beta a+\gamma_t$ and
$\xi_j := a_j b_j+\sum_t\tilde a_{jt}\bar\varepsilon_{jt}$.
The cluster-level quantities $(\tau_j,b_j,\{\varepsilon_{ijt}\})$ are i.i.d.\
across $j$, and $h_{jt}$, $\bar h_j$, $a_j$, $\tilde a_{jt}$ and
$\bar A_j$ are bounded functions of $\tau_j$.

\paragraph{Preliminary facts.}
(F1) $\EE[h_{jt}]=0$ for each $t$, hence $\EE[\bar h_j]=0$ and
$\EE[\tilde a_{jt}]=0$.
(F2) $\EE[h_{jt}A_{jt}]=\EE[A_{jt}]/g_t=1$.
(F3) $\EE[h_{jt}A_{js}]=\mathrm{Cov}(A_{jt},A_{js})/\{g_t(1-g_t)\}$, so
$\kappaM=\EE[\bar h_j\bar A_j]=T^{-2}\sum_{t,s}\mathrm{Cov}(A_{jt},A_{js})/\{g_t(1-g_t)\}$;
since $A_{jt}=\mathbf 1\{\tau_j\le t\}$ is monotone in $t$,
$\mathrm{Cov}(A_{jt},A_{js})=g_{t\wedge s}(1-g_{t\vee s})>0$, and $\kappaM>0$.
(F4) $\EE[h_{jt}^2]=1/g_t+1/(1-g_t)=\{g_t(1-g_t)\}^{-1}$ and
$\mathrm{Var}(a_j)=c_g^{-2}\sum_{t,s}\mathrm{Cov}(A_{jt},A_{js})>0$.
(F5) For any bounded function $w(\tau_j)$ and fold $v$, the average
$J^{-1}\sum_{j\in I_v}w(\tau_j)$ equals $(J_v/J)\,\EE[w]+O_P(J^{-1/2})$ and
is independent of all quantities computed from the clusters outside fold $v$.

\paragraph{Step 1: linearization of the training-fold least squares fit.}
Fix a fold $v$ and let $\mathcal T_v$ be the set of $J_{-v}=J-J_v$ training
clusters.  The regressors are an intercept, $A$ and the period indicators;
the intercept lies in the span of the period indicators, so by the
Frisch--Waugh--Lovell theorem the coefficient on $A$ is the coefficient in
the regression of $Y$ on $A$ after both are residualized on the period
indicators.  Within period $t$ the residual of $A_{jt}$ is
$A_{jt}-\bar A_t^{-v}$ with $\bar A_t^{-v} := J_{-v}^{-1}\sum_{j\in\mathcal T_v}A_{jt}$,
and the period effects $\gamma_t$ are removed by the residualization. On the event where stabilization is inactive,
\[
  \hat\beta^{(-v)}-\beta
  = \frac{\sum_{j\in\mathcal T_v}\sum_t m\,(A_{jt}-\bar A_t^{-v})(b_j+\bar\varepsilon_{jt})}
         {\sum_{j\in\mathcal T_v}\sum_t m\,(A_{jt}-\bar A_t^{-v})^2}.
\]
On the event where stabilization is inactive, the denominator equals $m J_{-v}\sum_t\bar A_t^{-v}(1-\bar A_t^{-v})
= mJ_{-v}\{c_g+O_P(J^{-1/2})\}$.  In the numerator, replace
$\bar A_t^{-v}$ by $g_t$: the difference is
$-m\sum_t(\bar A_t^{-v}-g_t)\sum_{j\in\mathcal T_v}(b_j+\bar\varepsilon_{jt})
= m\cdot O_P(J^{-1/2})\cdot O_P(J^{1/2}) = O_P(m)$, while
$\sum_t(A_{jt}-g_t)(b_j+\bar\varepsilon_{jt}) = c_g\,\xi_j$.  Hence
\begin{equation}
  \hat\beta^{(-v)}-\beta
  = \frac{c_g\sum_{j\in\mathcal T_v}\xi_j+O_P(1)}{J_{-v}\{c_g+O_P(J^{-1/2})\}}
  = \frac{1}{J_{-v}}\sum_{j\in\mathcal T_v}\xi_j+O_P(J^{-1}),
  \label{eq:s8_beta}
\end{equation}
using $J_{-v}^{-1}\sum_{j\in\mathcal T_v}\xi_j=O_P(J^{-1/2})$.  The period
and intercept coefficients are likewise $\sqrt{J}$-consistent, so
$d_t^{(-v)} := Q^0(0,t)-\hat Q^{(-v)}(0,t) = \gamma_t-\hat\gamma_t^{(-v)}-\hat\alpha^{(-v)}
= O_P(J^{-1/2})$ for each $t$, and
\begin{equation}
  Q^0(a,t)-\hat Q^{(-v)}(a,t) = (\beta-\hat\beta^{(-v)})\,a+d_t^{(-v)} .
  \label{eq:s8_Qerr}
\end{equation}

\paragraph{Step 2: the cross-fitted estimator (part (i)).}
Since $\hat Q^{(-v)}(1,t)-\hat Q^{(-v)}(0,t)=\hat\beta^{(-v)}$ and
$Y_{ijt}-\hat Q^{(-v(j))}(A_{jt},t) = b_j+\varepsilon_{ijt}+(\beta-\hat\beta^{(-v(j))})A_{jt}+d_t^{(-v(j))}$
by~\eqref{eq:s8_Qerr}, the estimator~\eqref{eq:base_est} of the main paper is
\[
  \Psihat = \sum_v\frac{J_v}{J}\hat\beta^{(-v)} + \frac1J\sum_j D_j + T_1 + T_2,
\]
with $D_j = \bar h_j b_j+(mT)^{-1}\sum_{i,t}h_{jt}\varepsilon_{ijt}$,
\[
  T_1 := \sum_v(\beta-\hat\beta^{(-v)})\,\frac1J\sum_{j\in I_v}m_j,
  \quad m_j := \frac1T\sum_t h_{jt}A_{jt},
  \qquad
  T_2 := \sum_v\sum_t d_t^{(-v)}\,\frac{1}{JT}\sum_{j\in I_v}h_{jt}.
\]
Here we used $N^{-1}\sum_{j,i,t}h_{jt}X_j = J^{-1}\sum_j(T^{-1}\sum_t h_{jt})X_j$
for any cluster-level $X_j$, since every $(j,t)$ cell contains $m$ records.
By (F2) and (F5), $J^{-1}\sum_{j\in I_v}m_j = J_v/J+O_P(J^{-1/2})$ and is
independent of $\hat\beta^{(-v)}$, so
$T_1=\sum_v(J_v/J)(\beta-\hat\beta^{(-v)})+O_P(J^{-1})$; by (F1) and (F5),
$T_2 = \sum_v\sum_t O_P(J^{-1/2})\,O_P(J^{-1/2}) = O_P(J^{-1})$.  The
$\hat\beta^{(-v)}$ terms cancel exactly:
$\sum_v(J_v/J)\hat\beta^{(-v)}+\sum_v(J_v/J)(\beta-\hat\beta^{(-v)})=\beta$.
Therefore $\Psihat=\beta+J^{-1}\sum_j D_j+O_P(J^{-1})$.  The $D_j$ are
i.i.d.\ with mean zero, and since $b_j$ and $\{\varepsilon_{ijt}\}$ are
independent of each other and of $\tau_j$,
\[
  \mathrm{Var}(D_j)=\sigma_b^2\,\EE[\bar h_j^2]
  +\frac{\sigma_\varepsilon^2}{(mT)^2}\sum_{i,t}\EE[h_{jt}^2]
  =\sigma_b^2\,\EE[\bar h_j^2]+\frac{\sigma_\varepsilon^2}{mT}\,\EE[h^2],
\]
with $\EE[h^2]=T^{-1}\sum_t\{g_t(1-g_t)\}^{-1}$ by (F4).  This is $\sigIF$.
The convergence $J\,\mathrm{Var}(\Psihat)\to\sigIF$ requires the
$O_P(J^{-1})$ term to have second moment $o(J^{-1})$; each such term is a
product of two centered averages (or a ratio of the form
in~\eqref{eq:s8_beta}), and under $\EE[b_j^4]+\EE[\varepsilon_{ijt}^4]<\infty$
the Marcinkiewicz--Zygmund inequality gives fourth moments of order $J^{-2}$
for the averages, so by the Cauchy--Schwarz inequality the products have
second moments of order $J^{-2}$.  For the ratio, write $q_M=\sum_t\bar A_t(1-\bar A_t)$ and $\tilde q_M=\max(q_M,c_g/2)$. Its inverse is deterministically bounded by $2/c_g$. Since the treatment indicators are bounded and independent across clusters, $q_M<c_g/2$ has exponentially small probability by Hoeffding's inequality. On its complement, expand $1/q_M$ around $1/c_g$; the error is bounded by a constant times $|q_M-c_g|$. The numerator centered average has fourth moment $O(M^{-2})$ under the stated fourth moments, and $q_M-c_g$ has fourth moment $O(M^{-2})$. Cauchy--Schwarz gives an $O(M^{-2})$ bound for the squared ratio remainder. On the exceptional event, condition on all treatment assignments: the stabilized coefficient has conditional second and fourth moments bounded by constants times moments of averages of $b_j$ and $\varepsilon_{ijt}$, plus a bounded signal term. Its contribution is exponentially small. Thus the remainders in both estimator expansions are $O_{L^2}(J^{-1})$, which justifies the variance and covariance limits.

For the plug-in sandwich of $\Psihat$, the cluster-level influence curve at
the fitted nuisances is
\[
  \widehat{\mathrm{IC}}_j
  = \bigl(\hat\beta^{(-v(j))}-\Psihat\bigr)+D_j+m_j\bigl(\beta-\hat\beta^{(-v(j))}\bigr)
    +\frac1T\sum_t h_{jt}d_t^{(-v(j))}
  = D_j+\rho_j,
\]
with $\max_j|\rho_j|=O_P(J^{-1/2})$ because $|m_j|$ and $|h_{jt}|$ are
bounded by $\max_t\{g_t(1-g_t)\}^{-1}$.  Writing $S^2(x)$ for the sample
variance of $x_1,\ldots,x_J$,
$S^2(D+\rho) = S^2(D)+2S(D,\rho)+S^2(\rho)$ with $S^2(\rho)\le\{J/(J-1)\}\max_j\rho_j^2=O_P(J^{-1})$
and $|S(D,\rho)|\le\{S^2(D)S^2(\rho)\}^{1/2}=O_P(J^{-1/2})$, while
$S^2(D)\to\mathrm{Var}(D_j)$ by the law of large numbers.  Since
$J\,\VarSand=\{(J-1)\}^{-1}\sum_j(\widehat{\mathrm{IC}}_j-\overline{\widehat{\mathrm{IC}}})^2 = S^2(D+\rho)$,
part~(i) follows.

\paragraph{Step 3: the own-cluster shift (part (ii)).}
By~\eqref{eq:s8_Qerr}, cluster $j$'s mean residual is
\[
  \bar r_j = b_j+\bar\varepsilon_j+(\beta-\hat\beta^{(-v(j))})\bar A_j+\bar d^{(-v(j))},
  \qquad \bar d^{(-v)} := T^{-1}\sum_t d_t^{(-v)},
\]
so the shift $\Delta_J := J^{-1}\sum_j\bar h_j\bar r_j$ equals
\[
  \frac1J\sum_j\bigl(\bar h_j b_j+\bar h_j\bar\varepsilon_j\bigr)
  +\sum_v(\beta-\hat\beta^{(-v)})\,\frac1J\sum_{j\in I_v}\bar h_j\bar A_j
  +\sum_v\bar d^{(-v)}\,\frac1J\sum_{j\in I_v}\bar h_j .
\]
By (F3) and (F5) the middle term is
$\kappaM\sum_v(J_v/J)(\beta-\hat\beta^{(-v)})+O_P(J^{-1})$, and by (F1) and
(F5) the last term is $O_P(J^{-1})$.  Inserting~\eqref{eq:s8_beta},
\[
  \sum_v\frac{J_v}{J}\bigl(\hat\beta^{(-v)}-\beta\bigr)
  = \frac1J\sum_k\xi_k\,w_k+O_P(J^{-1}),
  \qquad
  w_k := \sum_{v\ne v(k)}\frac{J_v}{J_{-v}} .
\]
With $|J_v-J/V|\le1$ we have $J_v/J_{-v}=(V-1)^{-1}\{1+O(J^{-1})\}$ for
every $v$, so $w_k=1+O(J^{-1})$ uniformly in $k$, and the right-hand side is
$J^{-1}\sum_k\xi_k+O_P(J^{-1})$.  Hence
\[
  \Delta_J = \frac1J\sum_j\bigl(\bar h_j b_j+\bar h_j\bar\varepsilon_j-\kappaM\,\xi_j\bigr)+O_P(J^{-1}),
\]
and, since $\Psihat_{\mathrm{own}}=\Psihat-\Delta_J$,
\[
  \Psihat_{\mathrm{own}}-\beta
  = \frac1J\sum_j\Bigl\{D_j-\bar h_j b_j-\bar h_j\bar\varepsilon_j
    +\kappaM a_j b_j+\kappaM\sum_t\tilde a_{jt}\bar\varepsilon_{jt}\Bigr\}+O_P(J^{-1}).
\]
The $\bar h_j b_j$ terms cancel. Convert the period means to individual sums using
\[
 \bar h_j\bar\varepsilon_j=\frac1{mT}\sum_{i,t}\bar h_j\varepsilon_{ijt},\qquad
 \kappaM\sum_t\tilde a_{jt}\bar\varepsilon_{jt}=\frac1{mT}\sum_{i,t}\kappaM T\tilde a_{jt}\varepsilon_{ijt}.
\]
The summand equals
\[
  \phi_j=\kappaM\,a_j b_j+\frac1{mT}\sum_{i,t}\bigl(h_{jt}-\bar h_j+\kappaM T\tilde a_{jt}\bigr)\varepsilon_{ijt}
  =\kappaM\,a_j b_j+\frac1{mT}\sum_{i,t}c_{jt}\varepsilon_{ijt}.
\]
The $\phi_j$ are i.i.d.\ with mean zero (by (F1) and independence) and
variance $\kappaM^2\sigma_b^2\EE[a_j^2]+\sigma_\varepsilon^2(mT)^{-2}\,m\sum_t\EE[c_{jt}^2]
=\kappaM^2\sigma_b^2\EE[a_j^2]+\sigma_\varepsilon^2\EE[c^2]/(mT)=\sigma^2_{\mathrm{own}}$.
Since the remainder is $O_P(J^{-1})=o_P(J^{-1/2})$, the central limit
theorem gives $\sqrt J(\Psihat_{\mathrm{own}}-\beta)\Rightarrow N(0,\sigma^2_{\mathrm{own}})$.
The expression for $\phi_j-D_j$ in Step 5 determines whether the two functions differ. They differ in the numerical configuration of the main paper; degenerate configurations can make them coincide.

\paragraph{Step 4: the plug-in sandwich of the own-cluster estimator (part (iii)).}
Since $\tilde Q_j(a,t)=\hat Q^{(-v(j))}(a,t)+\bar r_j$, we have
$\tilde Q_j(1,t)-\tilde Q_j(0,t)=\hat\beta^{(-v(j))}$ and
$Y_{ijt}-\tilde Q_j(A_{jt},t)=r_{ijt}-\bar r_j$ with
$r_{ijt}:=Y_{ijt}-\hat Q^{(-v(j))}(A_{jt},t)=b_j+\varepsilon_{ijt}+(\beta-\hat\beta^{(-v(j))})A_{jt}+d_t^{(-v(j))}$.
The cluster effect cancels in the difference:
\[
  r_{ijt}-\bar r_j=(\varepsilon_{ijt}-\bar\varepsilon_j)
  +(\beta-\hat\beta^{(-v(j))})(A_{jt}-\bar A_j)+\bigl(d_t^{(-v(j))}-\bar d^{(-v(j))}\bigr).
\]
Hence the plug-in influence curve is
\[
  \widehat{\mathrm{IC}}_j^{\mathrm{own}}
  = U_j+\rho_j',
  \qquad
  U_j := \frac1{mT}\sum_{i,t}(h_{jt}-\bar h_j)\varepsilon_{ijt},
\]
where we used $(mT)^{-1}\sum_{i,t}h_{jt}\bar\varepsilon_j=\bar h_j\bar\varepsilon_j
=(mT)^{-1}\sum_{i,t}\bar h_j\varepsilon_{ijt}$, and
\[
  \rho_j' := \bigl(\hat\beta^{(-v(j))}-\Psihat_{\mathrm{own}}\bigr)
  +(\beta-\hat\beta^{(-v(j))})\,\frac1T\sum_t h_{jt}(A_{jt}-\bar A_j)
  +\frac1T\sum_t h_{jt}\bigl(d_t^{(-v(j))}-\bar d^{(-v(j))}\bigr)
\]
satisfies $\max_j|\rho_j'|=O_P(J^{-1/2})$.  The argument of Step~2 gives
$J\,\VarSand(\Psihat_{\mathrm{own}})=S^2(U+\rho')\xrightarrow{P}\mathrm{Var}(U_j)$,
and
\begin{align*}
  \mathrm{Var}(U_j)&=\frac{\sigma_\varepsilon^2}{(mT)^2}\,m\sum_t\EE[(h_{jt}-\bar h_j)^2]\\
  &=\frac{\sigma_\varepsilon^2}{mT}\Bigl\{\frac1T\sum_t\EE[h_{jt}^2]-\EE[\bar h_j^2]\Bigr\}
  =\sigma^2_{\mathrm{sand}}.
\end{align*}
because $T^{-1}\sum_t(h_{jt}-\bar h_j)^2=T^{-1}\sum_t h_{jt}^2-\bar h_j^2$.
Combining with part~(ii),
\[
  \frac{\VarSand(\Psihat_{\mathrm{own}})}{\mathrm{Var}(\Psihat_{\mathrm{own}})}
  =\frac{J\,\VarSand(\Psihat_{\mathrm{own}})}{J\,\mathrm{Var}(\Psihat_{\mathrm{own}})}
  \xrightarrow{P}\frac{\sigma^2_{\mathrm{sand}}}{\sigma^2_{\mathrm{own}}}.
\]
When $\sigma_b^2>0$, the bound follows from $\sigma^2_{\mathrm{own}}\ge\kappaM^2\EE[a_j^2]\sigma_b^2$,
and the right-hand side of the bound is proportional to $1/m$, which gives
the limit as $m\to\infty$; by (F3) and (F4), $\kappaM^2\EE[a_j^2]>0$.

\paragraph{Step 5: the decomposition (part (iv)).}
With $\Dstar$ the influence function of part~(i), $\bar D_J=J^{-1}\sum_j D_j$
and $\Rrem=\Psihat_{\mathrm{own}}-\beta-\bar D_J=J^{-1}\sum_j(\phi_j-D_j)+O_P(J^{-1})$.
The $\phi_j-D_j$ are i.i.d.\ with mean zero, so under the fourth-moment
condition (which controls the $O_P(J^{-1})$ term as in Step~2)
$J\,\mathrm{Var}(\Rrem)\to\EE[(\phi_j-D_j)^2]$ and
$2J\,\mathrm{Cov}(\bar D_J,\Rrem)\to2\,\EE[D_j(\phi_j-D_j)]$.  Explicitly,
with $\phi_j-D_j=(\kappaM a_j-\bar h_j)b_j+(mT)^{-1}\sum_{i,t}(\kappaM T\tilde a_{jt}-\bar h_j)\varepsilon_{ijt}$,
\begin{align*}
  \EE[(\phi_j-D_j)^2]
  &=\sigma_b^2\,\EE[(\kappaM a_j-\bar h_j)^2]
    +\frac{\sigma_\varepsilon^2}{mT}\cdot\frac1T\sum_t\EE[(\kappaM T\tilde a_{jt}-\bar h_j)^2],\\
  \EE[D_j(\phi_j-D_j)]
  &=\sigma_b^2\,\EE[\bar h_j(\kappaM a_j-\bar h_j)]
    +\frac{\sigma_\varepsilon^2}{mT}\cdot\frac1T\sum_t\EE[h_{jt}(\kappaM T\tilde a_{jt}-\bar h_j)].
\end{align*}
For $g_t=t/6$ ($t=1,\ldots,5$), $m=40$, $\sigma_b^2=0.0337$ and $\sigma_\varepsilon^2=0.64$, these constants are $0.02182351$ and $-0.04283105$, respectively. Both contribute to the gap; their interpretation follows from the explicit first-order expansion. The script \texttt{R/stylized\_model.R} enumerates all six crossover patterns and independently simulates the estimators.

\paragraph{Step 6: fixed-partition jackknife for model (M).}
There are finitely many folds and treatment patterns. Collect each fold's averages of $A_t$, $A_tA_s$, $Y_t$, $A_tY_s$, and the products involving $h_t$ needed to evaluate the displayed estimators into a fixed-dimensional vector of empirical moments. The residual-period means $Y_t$ have finite fourth moments; treatment variables and $h_t$ are bounded. On an event with probability $1-O(Je^{-cJ})$, all full-sample and deletion training denominators exceed $c_g/2$. On this event the estimators are smooth rational functions of these moments, with derivatives bounded locally around their population values. A foldwise Taylor expansion after one deletion yields
\[
 \hat\Psi_{\mathrm{own}}^{(-j)}-\bar\Psi_{\mathrm{own},(-\cdot)}
 =-\frac{\phi_j-\bar\phi_J}{J-1}+q_j,
 \qquad J\sum_jq_j^2=o_P(1).
\]
To see the quadratic bound, the change in a fold mean is $O\{(1+\|Z_j\|)/J\}$ for its fixed-dimensional cluster moment vector $Z_j$. Taylor remainders have squared sum times $J$ bounded by $C J^{-3}\sum_j(1+\|Z_j\|^4)=O_P(J^{-2})$. Replacing sample derivatives by their limits adds $o_P(1)$ to the quadratic sum. Unequal fold sizes differ from $J/V$ by at most one; keeping their actual weights in the rational formula yields the same limiting derivative $\phi_j$ established in Step 3. The exceptional event has vanishing probability. The sample-variance law of large numbers and Cauchy--Schwarz now give $J\VarJK\to_P\mathrm{Var}(\phi_j)$. This argument is restricted to the fixed-dimensional model (M). \qed

\section{Stepped-Wedge Simulation Design and Results}\label{sec:sm:sim}
\subsection*{S9.1\; Generator and sequential targeting}
For each cluster, $b_j\sim N(0,\sigma_b^2)$ and $\tau_j\sim\mathrm{Uniform}\{2,\ldots,7\}$ independently. Put $A_{jt}=\mathbf1\{t\ge\tau_j\}$ and
\[
 \lambda_t=0.5\sin\{\pi(t-1)/6\}+0.3\{(t-1)/6\}^2.
\]
For each new individual draw $W_1\sim N(0,1)$, $W_2\sim\mathrm{Bernoulli}(0.4)$, $W_3\sim\mathrm{Uniform}(0,1)$ and independent errors $e_S\sim N(0,0.25)$ and $e_Y\sim N(0,0.64)$. Generate
\begin{align*}
 S&=0.8A+0.4W_1-0.3W_2+0.5\lambda_t+0.6b_j+e_S,\\
 Y&=-0.28A+0.5S+0.4W_1-0.2W_2+0.6W_3+0.8\lambda_t+b_j+e_Y,\\
 \Delta&=\mathbf1\{t\le\tau_j+3\}\,B,\\
 B\mid A,S,W,t&\sim\mathrm{Bernoulli}\!\left[\operatorname{expit}\{3-2.5\mathbf1(t\ge6)
 -0.8\mathbf1(t=7)+0.4S+0.3A-0.2W_3\}\right].
\end{align*}
There are 40 individuals in each of seven periods. Estimation retains periods 2 through 6 only. The ICC label is $\sigma_b^2/(\sigma_b^2+0.64)$; set $\sigma_b^2=0.64\rho/(1-\rho)$. It is a generator parameter, not the marginal correlation of the complete outcome or of the estimated residual contribution, both of which also involve other terms in the generator.

Let $m_S(a,W,t)=0.8a+0.4W_1-0.3W_2+0.5\lambda_t$ and $k_b=0.6\sigma_b^2/(0.36\sigma_b^2+0.25)$. The true marginal nuisance functions on these periods are
\begin{align*}
 Q_Y^0&=-0.28A+0.5S+0.4W_1-0.2W_2+0.6W_3+0.8\lambda_t
            +k_b\{S-m_S(A,W,t)\},\\
 Q_{\mathrm{int}}^0(a,W,t)&=-0.28a+0.5m_S(a,W,t)+0.4W_1-0.2W_2+0.6W_3+0.8\lambda_t,\\
 g_\Delta^0&=a_0(A,t)\operatorname{expit}\{3-2.5\mathbf1(t=6)+0.4S+0.3A-0.2W_3\}.
\end{align*}
Here $a_0=0.8$ for treated observations in period 6 and one otherwise, because among the five possible treated crossover times in that period only four satisfy administrative observation. The contrast $Q_{\mathrm{int}}^0(1)-Q_{\mathrm{int}}^0(0)$ is $0.12$ on every record.

For each fold, ordinary least squares fits $Q_Y$ on observed training outcomes. The integrated regression fits the predicted $Q_Y$ values on all training records using $A,W_1,W_2,W_3$ and period indicators. The censoring fit minimizes the negative Bernoulli log likelihood for $a_0\operatorname{expit}(x^\top\beta)$ plus $\frac12\sum_{r>1}\beta_r^2$; the intercept is unpenalized. Added independent Gaussian covariates enter the outcome and censoring models only. The optimizer must converge and its maximum normalized score must be at most $2\times10^{-6}$; failures stop and are logged rather than silently discarded. Outcome normal equations use a fixed $10^{-8}$ numerical ridge.

Write $H=h/\hat g_\Delta$. On the assembled predictions from the held-out folds, the fluctuation is
\[
 \hat\epsilon=\frac{\sum_{\Delta_i=1}H_i(Y_i-\hat Q_{Y,i})}{\sum_{\Delta_i=1}H_i^2},\qquad
 Q_{Y,i}^*=\hat Q_{Y,i}+\hat\epsilon H_i.
\]
With $H_{1,i}=1/g_{t_i}$ and $H_{0,i}=-1/(1-g_{t_i})$, compute
\[
 \hat\alpha_a=\frac{\sum_{A_i=a}H_{a,i}(Q_{Y,i}^*-\hat Q_{\mathrm{int},i}(a))}{\sum_{A_i=a}H_{a,i}^2},\quad
 Q_{\mathrm{int},i}^*(a)=\hat Q_{\mathrm{int},i}(a)+\hat\alpha_aH_{a,i}.
\]
Then $\hat\Psi=P_n\{Q_{\mathrm{int}}^*(1)-Q_{\mathrm{int}}^*(0)\}$ and each channel of the empirical score sums to zero. The data generator always contains untreated and treated records in the population, but finite training or evaluation samples can be deficient; the included numerical ridge defines the outcome fit, and the full run verifies finite fluctuation denominators and finite outputs in every reported replicate.

\paragraph{Correspondence in the regular regime.}
For a frozen nuisance function and known treatment probabilities, iterated expectation yields
\[
 P_0 d_{\eta,\psi}=\Psi_0-\psi+
 P_0\!\left[h\frac{g_\Delta-g_\Delta^0}{g_\Delta}(Q_Y-Q_Y^0)\right].
\]
The integrated-regression terms cancel because $\EE\{\mathbf1(A=a)/g_a\mid W,t\}=1$. The targeting equations yield an expansion with an empirical-process term and a weighted bilinear population remainder. At the truth this is exactly the reference score used for the oracle decomposition.

More explicitly, consider fixed-dimensional identifiable nuisance models with exact fits or numerical optimization errors negligible on the $J^{-1/2}$ scale. Suppose their coefficient errors are $O_P(J^{-1/2})$, score denominators converge to positive constants, and the local score classes have square-integrable envelopes. Then all three fluctuation coefficients are $O_P(J^{-1/2})$. The final nuisance errors remain $O_P(J^{-1/2})$, the population bilinear term is $O_P(J^{-1})$, and finite-dimensional stochastic equicontinuity makes the centered empirical-process difference $o_P(J^{-1/2})$. Consequently,\[\hat\Psi-\Psi_0=J^{-1}\sum_jD_j^0+o_P(J^{-1/2}).\] The normal/Gaussian-linear generator has finite local inverse-probability moments; the vanishing probability guard is inactive on compact coefficient neighborhoods with probability tending to one. The expansion uses fixed-dimensional regularity conditions; variance convergence requires corresponding moment control.

\paragraph{Own-cluster comparisons.}
After cross-fitted prediction, add the cluster's mean observed outcome residual to $Q_Y,Q_{\mathrm{int}}(1),Q_{\mathrm{int}}(0)$. In the joint adjustment also add
$\sum_{i\in j}(\Delta_i-\hat g_i)/\sum_{i\in j}\hat g_i(1-\hat g_i)$ to the logit of the fitted censoring probability. These are deliberately evaluation-data-dependent modifications. The shrunken outcome variant multiplies its mean residual by $n_{Y,j}\sigma_b^2/(n_{Y,j}\sigma_b^2+0.64)$; this is a prespecified shrinkage rule applied to the own-cluster residual. All variants use the same targeting equations.

\subsection*{S9.2\; Empirical plug-in diagnostics}
We compare the final targeted score with the initial cross-fitted score, and replace the outcome or censoring predictions by their known truths. For the nuisance-isolation comparisons the integrated regression is held at its truth. At the cluster level, if $e_j=d_{\hat\eta,j}^{\mathrm{initial}}-D_j^0$, the sample identity is
\[
 s^2(d_{\hat\eta}^{\mathrm{initial}})=s^2(D^0)+s^2(e)+2s(D^0,e).
\]
Both the error variance and its cross covariance are reported. The empirical variance $s^2(e)$ is centered and is not the uncentered, fold-conditional population quantity $\Lambda_n$ in Proposition~\ref{prop:plugin}. The mean squared difference between the final targeted and pre-targeting scores measures the magnitude of the fluctuation update.
\begin{table}[p]
\centering
\scriptsize
\caption{Score variances multiplied by $10^3$. Targeted is $J$ times the reported final variance; other columns are sample variances of cluster scores. Q-only and g-only hold the integrated regression at the truth.}
\label{tab:arxiv_s3}
\setlength{\tabcolsep}{3pt}
\resizebox{\linewidth}{!}{%
\begin{tabular}{rrrrrrrr}
\toprule
$J$ & $\kappa$ & $R$ & Targeted & Initial & Q-only & g-only & Oracle \\
\midrule
10 & 0.10 & 300 & 1317.6 & 74709.0 & 23.4 & 117727.4 & 23.2 \\
30 & 0.00 & 150 & 23.1 & 23.4 & 22.2 & 22.4 & 22.2 \\
30 & 0.05 & 150 & 33.3 & 36.0 & 22.4 & 33.9 & 22.2 \\
30 & 0.10 & 150 & 1722.3 & 10562.6 & 22.6 & 15494.1 & 22.2 \\
\bottomrule
\end{tabular}
}
\end{table}
\begin{table}[p]
\centering
\scriptsize
\caption{Empirical score-error variance, twice its cross covariance, and mean squared targeting change, all multiplied by $10^3$.}
\label{tab:arxiv_s4}
\setlength{\tabcolsep}{3pt}
\resizebox{\linewidth}{!}{%
\begin{tabular}{rrrrr}
\toprule
$J$ & $\kappa$ & $s^2(e)$ & $2s(D^0,e)$ & Targeting change \\
\midrule
10 & 0.10 & 74643.8 & 42.0 & 58989.8 \\
30 & 0.00 & 1.4 & -0.3 & 0.1 \\
30 & 0.05 & 10.7 & 3.1 & 1.0 \\
30 & 0.10 & 10513.9 & 26.5 & 8032.0 \\
\bottomrule
\end{tabular}
}
\end{table}

\subsection*{S9.3\; Results and Monte Carlo uncertainty}
The manifest contains 40 cells and 7,560 unique replicates. Jackknife computations are performed in 2,860 prespecified replicates, retaining each replicate's fold labels and noise covariates. Each deletion refits every nuisance and targeting step. Replicate seeds are $2024+r$ for data and $7r$ for the noise covariates. There is no outcome-based stopping or replicate filtering.

Monte Carlo variances use all replicates in a cell. Intervals for decomposition terms and ordinary variance ratios use 1,000 paired bootstrap resamples of replicate indices. When the jackknife is available in only a subset, the bootstrap resamples jointly within the jackknife and non-jackknife strata, preserving the dependence between the estimated-variance numerator and the empirical-variance denominator. Coverage intervals are two-sided exact binomial intervals. Missing jackknife entries denote computations absent from the manifest. The tables summarize the finite-sample performance across the simulated configurations.
\begin{table}[p]
\centering
\scriptsize
\caption{Fixed and growing dimensions, ICC 0: residual decomposition. Parentheses give 95\% Monte Carlo intervals, except bias, which reports its Monte Carlo standard error. $V_R$ and $C$ are multiplied by $10^3$.}
\label{tab:arxiv_s5}
\setlength{\tabcolsep}{3pt}
\resizebox{\linewidth}{!}{%
\begin{tabular}{llrrrrr}
\toprule
Fit/$\kappa$ & $J$ & $R$ & Bias (s.e.) & $V_R$ (interval) & $C$ (interval) & Share (interval) \\
\midrule
0 & 10 & 300 & -0.003 (0.003) & 5.2 (4.2, 6.3) & 2.4 (-0.3, 4.8) & 0.24 (0.19, 0.30) \\
0 & 30 & 200 & -0.000 (0.002) & 1.2 (1.0, 1.5) & -0.6 (-1.9, 0.7) & 0.06 (0.04, 0.08) \\
0 & 100 & 100 & -0.001 (0.001) & 0.4 (0.3, 0.6) & -0.6 (-2.3, 0.8) & 0.02 (0.01, 0.03) \\
0.05 & 10 & 300 & -0.003 (0.003) & 7.6 (6.3, 9.0) & 2.0 (-0.7, 4.9) & 0.35 (0.27, 0.44) \\
0.05 & 30 & 200 & -0.000 (0.002) & 4.8 (3.8, 5.6) & -0.8 (-3.5, 1.6) & 0.23 (0.17, 0.30) \\
0.05 & 100 & 100 & -0.002 (0.001) & 4.2 (3.1, 5.3) & -4.0 (-7.3, -0.8) & 0.19 (0.13, 0.28) \\
0.10 & 10 & 300 & -0.003 (0.003) & 9.6 (8.0, 11.6) & -0.6 (-4.3, 3.0) & 0.44 (0.35, 0.55) \\
0.10 & 30 & 150 & -0.001 (0.002) & 9.5 (7.6, 11.6) & -8.1 (-12.6, -3.4) & 0.49 (0.36, 0.65) \\
0.10 & 100 & 60 & +0.000 (0.002) & 4.6 (3.3, 5.9) & -2.7 (-7.6, 2.3) & 0.20 (0.12, 0.34) \\
\bottomrule
\end{tabular}
}
\end{table}

\begin{table}[p]
\centering
\scriptsize
\caption{Fixed and growing dimensions, ICC 0: variance ratios and coverage, with 95\% Monte Carlo intervals.}
\label{tab:arxiv_s6}
\setlength{\tabcolsep}{3pt}
\resizebox{\linewidth}{!}{%
\begin{tabular}{lllrll}
\toprule
Fit/$\kappa$ & $J$ & Method & $R_{\mathrm{method}}$ & Variance ratio & Coverage \\
\midrule
0 & 10 & Plug-in & 300 & 0.88 (0.74, 1.06) & 0.903 (0.864, 0.934) \\
0 & 10 & Oracle & 300 & 0.78 (0.67, 0.93) & 0.913 (0.876, 0.943) \\
0 & 10 & Jackknife & 300 & 1.20 (1.02, 1.43) & 0.937 (0.903, 0.961) \\
0 & 30 & Plug-in & 200 & 1.08 (0.90, 1.34) & 0.950 (0.910, 0.976) \\
0 & 30 & Oracle & 200 & 1.04 (0.87, 1.30) & 0.955 (0.916, 0.979) \\
0 & 30 & Jackknife & 100 & 1.13 (0.93, 1.40) & 1.000 (0.964, 1.000) \\
0 & 100 & Plug-in & 100 & 1.08 (0.86, 1.46) & 0.960 (0.901, 0.989) \\
0 & 100 & Oracle & 100 & 1.06 (0.84, 1.43) & 0.960 (0.901, 0.989) \\
0.05 & 10 & Plug-in & 300 & 1.04 (0.86, 1.25) & 0.907 (0.868, 0.937) \\
0.05 & 10 & Oracle & 300 & 0.73 (0.62, 0.88) & 0.897 (0.857, 0.929) \\
0.05 & 10 & Jackknife & 300 & 1.43 (1.19, 1.72) & 0.967 (0.940, 0.984) \\
0.05 & 30 & Plug-in & 200 & 1.34 (1.12, 1.64) & 0.950 (0.910, 0.976) \\
0.05 & 30 & Oracle & 200 & 0.91 (0.76, 1.11) & 0.925 (0.879, 0.957) \\
0.05 & 30 & Jackknife & 100 & 1.44 (1.21, 1.75) & 0.990 (0.946, 1.000) \\
0.05 & 100 & Plug-in & 100 & 1.64 (1.30, 2.18) & 0.990 (0.946, 1.000) \\
0.05 & 100 & Oracle & 100 & 1.04 (0.83, 1.38) & 0.960 (0.901, 0.989) \\
0.10 & 10 & Plug-in & 300 & 42.55 (8.41, 120.40) & 0.990 (0.971, 0.998) \\
0.10 & 10 & Oracle & 300 & 0.75 (0.64, 0.89) & 0.900 (0.860, 0.932) \\
0.10 & 10 & Jackknife & 300 & 6.15 (2.36, 11.46) & 0.970 (0.944, 0.986) \\
0.10 & 30 & Plug-in & 150 & 82.25 (36.12, 156.73) & 1.000 (0.976, 1.000) \\
0.10 & 30 & Oracle & 150 & 1.06 (0.84, 1.43) & 0.933 (0.881, 0.968) \\
0.10 & 30 & Jackknife & 100 & 31.01 (1.99, 105.30) & 1.000 (0.964, 1.000) \\
0.10 & 100 & Plug-in & 60 & 83.81 (49.44, 156.69) & 1.000 (0.940, 1.000) \\
0.10 & 100 & Oracle & 60 & 0.94 (0.67, 1.54) & 0.933 (0.838, 0.982) \\
\bottomrule
\end{tabular}
}
\end{table}

\begin{table}[p]
\centering
\scriptsize
\caption{Fixed and growing dimensions, ICC 0.013: residual decomposition. Parentheses give 95\% Monte Carlo intervals, except bias, which reports its Monte Carlo standard error. $V_R$ and $C$ are multiplied by $10^3$.}
\label{tab:arxiv_s7}
\setlength{\tabcolsep}{3pt}
\resizebox{\linewidth}{!}{%
\begin{tabular}{llrrrrr}
\toprule
Fit/$\kappa$ & $J$ & $R$ & Bias (s.e.) & $V_R$ (interval) & $C$ (interval) & Share (interval) \\
\midrule
0 & 10 & 300 & -0.003 (0.005) & 13.2 (10.7, 16.1) & -2.6 (-9.6, 3.5) & 0.22 (0.17, 0.28) \\
0 & 30 & 200 & +0.002 (0.003) & 4.2 (3.2, 5.2) & -0.6 (-4.9, 4.2) & 0.07 (0.05, 0.09) \\
0 & 100 & 100 & -0.001 (0.003) & 1.0 (0.7, 1.3) & 0.9 (-2.1, 3.7) & 0.02 (0.01, 0.02) \\
0.05 & 10 & 300 & -0.005 (0.005) & 15.5 (12.6, 18.5) & -1.1 (-8.1, 6.2) & 0.25 (0.19, 0.32) \\
0.05 & 30 & 200 & +0.002 (0.003) & 8.6 (7.0, 10.3) & -1.8 (-7.7, 4.5) & 0.14 (0.11, 0.19) \\
0.05 & 100 & 60 & -0.002 (0.004) & 6.0 (4.1, 7.8) & 4.7 (-5.2, 13.4) & 0.09 (0.06, 0.14) \\
0.10 & 10 & 300 & -0.005 (0.005) & 16.2 (13.4, 19.3) & -5.7 (-13.2, 1.7) & 0.26 (0.21, 0.34) \\
0.10 & 30 & 150 & +0.002 (0.004) & 8.5 (6.6, 10.8) & -4.9 (-11.0, 1.7) & 0.14 (0.10, 0.20) \\
0.10 & 100 & 60 & -0.001 (0.004) & 8.4 (5.6, 11.1) & 3.1 (-7.5, 14.2) & 0.13 (0.08, 0.22) \\
\bottomrule
\end{tabular}
}
\end{table}

\begin{table}[p]
\centering
\scriptsize
\caption{Fixed and growing dimensions, ICC 0.013: variance ratios and coverage, with 95\% Monte Carlo intervals.}
\label{tab:arxiv_s8}
\setlength{\tabcolsep}{3pt}
\resizebox{\linewidth}{!}{%
\begin{tabular}{lllrll}
\toprule
Fit/$\kappa$ & $J$ & Method & $R_{\mathrm{method}}$ & Variance ratio & Coverage \\
\midrule
0 & 10 & Plug-in & 300 & 0.94 (0.78, 1.18) & 0.920 (0.883, 0.948) \\
0 & 10 & Oracle & 300 & 0.85 (0.71, 1.05) & 0.943 (0.911, 0.967) \\
0 & 30 & Plug-in & 200 & 1.08 (0.89, 1.35) & 0.960 (0.923, 0.983) \\
0 & 30 & Oracle & 200 & 1.04 (0.86, 1.29) & 0.960 (0.923, 0.983) \\
0 & 100 & Plug-in & 100 & 1.01 (0.81, 1.35) & 0.940 (0.874, 0.978) \\
0 & 100 & Oracle & 100 & 1.01 (0.80, 1.34) & 0.940 (0.874, 0.978) \\
0.05 & 10 & Plug-in & 300 & 1.02 (0.84, 1.26) & 0.933 (0.899, 0.959) \\
0.05 & 10 & Oracle & 300 & 0.81 (0.68, 0.97) & 0.940 (0.907, 0.964) \\
0.05 & 30 & Plug-in & 200 & 1.24 (1.02, 1.56) & 0.960 (0.923, 0.983) \\
0.05 & 30 & Oracle & 200 & 0.99 (0.82, 1.24) & 0.955 (0.916, 0.979) \\
0.05 & 30 & Jackknife & 100 & 1.28 (1.04, 1.58) & 0.970 (0.915, 0.994) \\
0.05 & 100 & Plug-in & 60 & 1.11 (0.83, 1.63) & 0.967 (0.885, 0.996) \\
0.05 & 100 & Oracle & 60 & 0.90 (0.67, 1.33) & 0.950 (0.861, 0.990) \\
0.10 & 10 & Plug-in & 300 & 7.97 (5.36, 11.45) & 0.987 (0.966, 0.996) \\
0.10 & 10 & Oracle & 300 & 0.85 (0.72, 1.03) & 0.943 (0.911, 0.967) \\
0.10 & 30 & Plug-in & 150 & 21.04 (12.90, 34.25) & 1.000 (0.976, 1.000) \\
0.10 & 30 & Oracle & 150 & 1.06 (0.84, 1.38) & 0.947 (0.898, 0.977) \\
0.10 & 100 & Plug-in & 60 & 25.32 (16.33, 44.04) & 1.000 (0.940, 1.000) \\
0.10 & 100 & Oracle & 60 & 0.90 (0.66, 1.38) & 0.933 (0.838, 0.982) \\
\bottomrule
\end{tabular}
}
\end{table}

\begin{table}[p]
\centering
\scriptsize
\caption{Fixed and growing dimensions, ICC 0.05: residual decomposition. Parentheses give 95\% Monte Carlo intervals, except bias, which reports its Monte Carlo standard error. $V_R$ and $C$ are multiplied by $10^3$.}
\label{tab:arxiv_s9}
\setlength{\tabcolsep}{3pt}
\resizebox{\linewidth}{!}{%
\begin{tabular}{llrrrrr}
\toprule
Fit/$\kappa$ & $J$ & $R$ & Bias (s.e.) & $V_R$ (interval) & $C$ (interval) & Share (interval) \\
\midrule
0 & 10 & 300 & -0.005 (0.008) & 42.7 (33.9, 52.1) & -8.8 (-31.3, 12.9) & 0.24 (0.18, 0.31) \\
0 & 30 & 200 & +0.004 (0.006) & 13.1 (9.9, 16.5) & -1.2 (-15.1, 12.7) & 0.07 (0.05, 0.10) \\
0 & 100 & 100 & -0.003 (0.004) & 3.0 (2.0, 4.1) & 3.0 (-5.3, 11.3) & 0.02 (0.01, 0.03) \\
0.05 & 10 & 300 & -0.007 (0.009) & 44.2 (35.4, 54.3) & -6.5 (-27.6, 14.8) & 0.24 (0.19, 0.32) \\
0.05 & 30 & 200 & +0.004 (0.006) & 16.6 (12.7, 20.5) & -2.8 (-17.0, 12.2) & 0.09 (0.07, 0.12) \\
0.05 & 60 & 100 & +0.000 (0.006) & 10.4 (6.5, 15.0) & 8.8 (-9.1, 27.3) & 0.06 (0.04, 0.09) \\
0.05 & 100 & 60 & -0.004 (0.006) & 6.7 (4.5, 8.8) & 15.4 (1.2, 31.0) & 0.04 (0.02, 0.06) \\
0.10 & 10 & 300 & -0.008 (0.008) & 38.5 (31.1, 46.6) & -13.7 (-34.4, 7.3) & 0.21 (0.16, 0.28) \\
0.10 & 30 & 200 & +0.003 (0.006) & 15.3 (12.0, 19.4) & -13.2 (-26.1, 0.3) & 0.08 (0.06, 0.12) \\
0.10 & 50 & 60 & +0.006 (0.009) & 11.7 (7.6, 15.6) & -16.3 (-42.0, 6.3) & 0.05 (0.03, 0.09) \\
0.10 & 100 & 60 & -0.003 (0.006) & 10.7 (7.3, 14.5) & 10.9 (-8.0, 30.1) & 0.06 (0.04, 0.10) \\
\bottomrule
\end{tabular}
}
\end{table}

\begin{table}[p]
\centering
\scriptsize
\caption{Fixed and growing dimensions, ICC 0.05: variance ratios and coverage, with 95\% Monte Carlo intervals.}
\label{tab:arxiv_s10}
\setlength{\tabcolsep}{3pt}
\resizebox{\linewidth}{!}{%
\begin{tabular}{lllrll}
\toprule
Fit/$\kappa$ & $J$ & Method & $R_{\mathrm{method}}$ & Variance ratio & Coverage \\
\midrule
0 & 10 & Plug-in & 300 & 0.92 (0.76, 1.14) & 0.907 (0.868, 0.937) \\
0 & 10 & Oracle & 300 & 0.85 (0.71, 1.03) & 0.937 (0.903, 0.961) \\
0 & 30 & Plug-in & 200 & 1.05 (0.88, 1.30) & 0.935 (0.891, 0.965) \\
0 & 30 & Oracle & 200 & 1.01 (0.85, 1.25) & 0.950 (0.910, 0.976) \\
0 & 100 & Plug-in & 100 & 1.12 (0.88, 1.49) & 0.960 (0.901, 0.989) \\
0 & 100 & Oracle & 100 & 1.11 (0.87, 1.46) & 0.960 (0.901, 0.989) \\
0.05 & 10 & Plug-in & 300 & 0.99 (0.83, 1.22) & 0.917 (0.879, 0.945) \\
0.05 & 10 & Oracle & 300 & 0.83 (0.71, 1.01) & 0.940 (0.907, 0.964) \\
0.05 & 30 & Plug-in & 200 & 1.16 (0.97, 1.44) & 0.945 (0.904, 0.972) \\
0.05 & 30 & Oracle & 200 & 1.00 (0.84, 1.23) & 0.945 (0.904, 0.972) \\
0.05 & 60 & Plug-in & 100 & 1.09 (0.89, 1.42) & 0.970 (0.915, 0.994) \\
0.05 & 60 & Oracle & 100 & 0.95 (0.77, 1.23) & 0.960 (0.901, 0.989) \\
0.05 & 100 & Plug-in & 60 & 1.17 (0.90, 1.65) & 0.983 (0.911, 1.000) \\
0.05 & 100 & Oracle & 60 & 1.03 (0.80, 1.45) & 0.950 (0.861, 0.990) \\
0.10 & 10 & Plug-in & 300 & 5.91 (3.32, 10.04) & 0.957 (0.927, 0.977) \\
0.10 & 10 & Oracle & 300 & 0.89 (0.76, 1.07) & 0.947 (0.915, 0.969) \\
0.10 & 30 & Plug-in & 200 & 297.35 (7.23, 984.42) & 1.000 (0.982, 1.000) \\
0.10 & 30 & Oracle & 200 & 1.06 (0.90, 1.32) & 0.960 (0.923, 0.983) \\
0.10 & 50 & Plug-in & 60 & 7.25 (4.40, 12.86) & 1.000 (0.940, 1.000) \\
0.10 & 50 & Oracle & 60 & 0.91 (0.67, 1.45) & 0.900 (0.795, 0.962) \\
0.10 & 100 & Plug-in & 60 & 10.91 (7.30, 17.52) & 1.000 (0.940, 1.000) \\
0.10 & 100 & Oracle & 60 & 1.04 (0.80, 1.51) & 0.950 (0.861, 0.990) \\
\bottomrule
\end{tabular}
}
\end{table}

\begin{table}[p]
\centering
\scriptsize
\caption{Own-cluster adjustments, ICC 0.05: residual decomposition. Parentheses give 95\% Monte Carlo intervals, except bias, which reports its Monte Carlo standard error. $V_R$ and $C$ are multiplied by $10^3$.}
\label{tab:arxiv_s11}
\setlength{\tabcolsep}{3pt}
\resizebox{\linewidth}{!}{%
\begin{tabular}{llrrrrr}
\toprule
Fit/$\kappa$ & $J$ & $R$ & Bias (s.e.) & $V_R$ (interval) & $C$ (interval) & Share (interval) \\
\midrule
Joint & 10 & 300 & -0.002 (0.006) & 37.4 (31.0, 44.1) & -95.0 (-117.7, -73.9) & 0.21 (0.17, 0.26) \\
Joint & 30 & 200 & +0.002 (0.004) & 23.7 (19.5, 28.0) & -96.0 (-116.6, -74.4) & 0.13 (0.11, 0.15) \\
Joint & 100 & 100 & -0.001 (0.003) & 15.8 (12.0, 20.1) & -72.4 (-99.2, -47.5) & 0.09 (0.07, 0.11) \\
Outcome & 10 & 300 & -0.003 (0.006) & 37.8 (31.2, 44.4) & -95.9 (-116.2, -75.6) & 0.21 (0.17, 0.25) \\
Outcome & 30 & 200 & +0.001 (0.004) & 23.8 (19.3, 27.9) & -96.1 (-118.4, -73.7) & 0.13 (0.11, 0.15) \\
Outcome & 100 & 100 & -0.001 (0.003) & 15.7 (11.9, 19.3) & -73.1 (-97.4, -49.0) & 0.09 (0.07, 0.11) \\
Shrunken & 10 & 300 & -0.003 (0.007) & 35.5 (29.2, 42.1) & -87.2 (-107.4, -67.1) & 0.20 (0.16, 0.24) \\
Shrunken & 30 & 200 & +0.002 (0.004) & 20.7 (17.1, 24.7) & -86.8 (-107.2, -67.2) & 0.11 (0.09, 0.14) \\
Shrunken & 100 & 100 & -0.002 (0.004) & 13.0 (9.9, 16.2) & -65.6 (-87.6, -43.3) & 0.07 (0.06, 0.09) \\
\bottomrule
\end{tabular}
}
\end{table}

\begin{table}[p]
\centering
\scriptsize
\caption{Own-cluster adjustments, ICC 0.05: variance ratios and coverage, with 95\% Monte Carlo intervals.}
\label{tab:arxiv_s12}
\setlength{\tabcolsep}{3pt}
\resizebox{\linewidth}{!}{%
\begin{tabular}{lllrll}
\toprule
Fit/$\kappa$ & $J$ & Method & $R_{\mathrm{method}}$ & Variance ratio & Coverage \\
\midrule
Joint & 10 & Plug-in & 300 & 0.18 (0.16, 0.23) & 0.633 (0.576, 0.688) \\
Joint & 10 & Oracle & 300 & 1.48 (1.25, 1.80) & 0.977 (0.953, 0.991) \\
Joint & 10 & Jackknife & 300 & 1.23 (1.01, 1.53) & 0.953 (0.923, 0.974) \\
Joint & 30 & Plug-in & 200 & 0.23 (0.19, 0.28) & 0.645 (0.574, 0.711) \\
Joint & 30 & Oracle & 200 & 1.78 (1.49, 2.19) & 0.990 (0.964, 0.999) \\
Joint & 30 & Jackknife & 200 & 1.15 (0.96, 1.42) & 0.970 (0.936, 0.989) \\
Joint & 100 & Plug-in & 100 & 0.21 (0.17, 0.28) & 0.590 (0.487, 0.687) \\
Joint & 100 & Oracle & 100 & 1.69 (1.35, 2.22) & 0.990 (0.946, 1.000) \\
Joint & 100 & Jackknife & 60 & 1.05 (0.84, 1.36) & 0.967 (0.885, 0.996) \\
Outcome & 10 & Plug-in & 300 & 0.18 (0.15, 0.22) & 0.630 (0.573, 0.685) \\
Outcome & 10 & Oracle & 300 & 1.48 (1.23, 1.79) & 0.973 (0.948, 0.988) \\
Outcome & 10 & Jackknife & 300 & 1.22 (1.02, 1.51) & 0.957 (0.927, 0.977) \\
Outcome & 30 & Plug-in & 200 & 0.22 (0.19, 0.28) & 0.630 (0.559, 0.697) \\
Outcome & 30 & Oracle & 200 & 1.78 (1.49, 2.21) & 0.990 (0.964, 0.999) \\
Outcome & 30 & Jackknife & 200 & 1.14 (0.95, 1.41) & 0.965 (0.929, 0.986) \\
Outcome & 100 & Plug-in & 100 & 0.21 (0.17, 0.28) & 0.600 (0.497, 0.697) \\
Outcome & 100 & Oracle & 100 & 1.70 (1.37, 2.22) & 0.990 (0.946, 1.000) \\
Shrunken & 10 & Plug-in & 300 & 0.22 (0.19, 0.27) & 0.667 (0.610, 0.720) \\
Shrunken & 10 & Oracle & 300 & 1.41 (1.19, 1.71) & 0.970 (0.944, 0.986) \\
Shrunken & 10 & Jackknife & 300 & 1.23 (1.02, 1.53) & 0.953 (0.923, 0.974) \\
Shrunken & 30 & Plug-in & 200 & 0.28 (0.24, 0.35) & 0.705 (0.637, 0.767) \\
Shrunken & 30 & Oracle & 200 & 1.68 (1.42, 2.09) & 0.990 (0.964, 0.999) \\
Shrunken & 30 & Jackknife & 200 & 1.14 (0.96, 1.39) & 0.965 (0.929, 0.986) \\
Shrunken & 100 & Plug-in & 100 & 0.27 (0.22, 0.35) & 0.680 (0.579, 0.770) \\
Shrunken & 100 & Oracle & 100 & 1.63 (1.32, 2.13) & 0.990 (0.946, 1.000) \\
\bottomrule
\end{tabular}
}
\end{table}

\begin{table}[p]
\centering
\scriptsize
\caption{Joint own-cluster adjustments, ICC 0.20: residual decomposition. Parentheses give 95\% Monte Carlo intervals, except bias, which reports its Monte Carlo standard error. $V_R$ and $C$ are multiplied by $10^3$.}
\label{tab:arxiv_s13}
\setlength{\tabcolsep}{3pt}
\resizebox{\linewidth}{!}{%
\begin{tabular}{llrrrrr}
\toprule
Fit/$\kappa$ & $J$ & $R$ & Bias (s.e.) & $V_R$ (interval) & $C$ (interval) & Share (interval) \\
\midrule
Joint & 10 & 300 & -0.006 (0.013) & 149.9 (119.5, 180.5) & -417.3 (-506.4, -331.1) & 0.19 (0.15, 0.23) \\
Joint & 30 & 200 & +0.005 (0.009) & 81.1 (65.7, 96.2) & -399.1 (-481.6, -310.9) & 0.10 (0.09, 0.12) \\
\bottomrule
\end{tabular}
}
\end{table}

\begin{table}[p]
\centering
\scriptsize
\caption{Joint own-cluster adjustments, ICC 0.20: variance ratios and coverage, with 95\% Monte Carlo intervals.}
\label{tab:arxiv_s14}
\setlength{\tabcolsep}{3pt}
\resizebox{\linewidth}{!}{%
\begin{tabular}{lllrll}
\toprule
Fit/$\kappa$ & $J$ & Method & $R_{\mathrm{method}}$ & Variance ratio & Coverage \\
\midrule
Joint & 10 & Plug-in & 300 & 0.19 (0.16, 0.22) & 0.590 (0.532, 0.646) \\
Joint & 10 & Oracle & 300 & 1.50 (1.27, 1.79) & 0.977 (0.953, 0.991) \\
Joint & 30 & Plug-in & 200 & 0.26 (0.22, 0.32) & 0.680 (0.611, 0.744) \\
Joint & 30 & Oracle & 200 & 1.74 (1.45, 2.15) & 0.990 (0.964, 0.999) \\
\bottomrule
\end{tabular}
}
\end{table}

\section{Gaussian Boundary Example and Its Jackknife}
\label{sec:sm:gaussian}
For each $n$, observe independent triples $(X_i,U_i,V_i)$, with $X_i\sim N(0,\sigma^2)$ and $U_i,V_i\sim N_{p_n}(0,I)$, mutually independent. In the Gaussian mean model the target is $\mu_X+\mu_U^\top\mu_V$, which is zero here, and its efficient influence function at zero is $D_i^*=X_i$. The plug-in estimator is
\[
 \hat\theta_n=\bar X+\bar U^\top\bar V,\qquad R_n=\bar U^\top\bar V.
\]
Exactly, $\EE R_n=0$, $\mathrm{Var}(R_n)=p_n/n^2$, and $\mathrm{Cov}(\bar X,R_n)=0$. If $p_n/n\to\kappaG>0$, then
\[
 \sqrt n\hat\theta_n\Rightarrow N(0,\sigma^2+\kappaG),\qquad
 n\mathrm{Var}(\hat\theta_n)\to\sigma^2+\kappaG.
\]
Equivalently, $R_n=\EE_Z\{(\bar U^\top Z)(\bar V^\top Z)\}$ for a fresh standard normal $Z$, so this is a bilinear $L^2$ inner product. The nuisance mean errors have non-vanishing Euclidean norms. The construction uses a triangular Gaussian mean model with $p_n/n\to\kappaG$. 

At zero means the differential of $\mu_X+\mu_U^\top\mu_V$ is $d\mu_X$, so $X_i$ is the efficient influence function in the Gaussian mean model. Write $Z=\sqrt n\bar U$ and $W=\sqrt n\bar V$. Their coordinates are independent standard normals, and
\[
 R_n=n^{-1}\sum_{\ell=1}^{p_n}Z_\ell W_\ell.
\]
The summands have mean zero, variance one, and fourth moment nine. Therefore $\mathrm{Var}(R_n)=p_n/n^2$ exactly. The Lyapunov ratio is $9/p_n\to0$, so $\sqrt nR_n\Rightarrow N(0,\kappaG)$ if $p_n/n\to\kappaG>0$. Independence from $\bar X$ proves the total Gaussian limit and its moment convergence.

For deletion keep $p_n$ fixed at its full-sample value. Put $u=\bar U$, $v=\bar V$, $R=u^\top v$, and $\psi_i=u^\top V_i+v^\top U_i$. Direct algebra gives
\[
 b_i=R^{(-i)}-R=
 \frac{(2n-1)R-nu^\top V_i-nv^\top U_i+U_i^\top V_i}{(n-1)^2}.
\]
Hence $n^{-1}\sum_i(nb_i+\psi_i)^2=o_P(1)$: the constant term $2R$ vanishes, the coefficient discrepancy for $\psi_i$ is $O(n^{-1})$, and
\[
 \EE\left\{\frac1n\sum_i\left(\frac{U_i^\top V_i}{n}\right)^2\right\}
 =\frac{p_n}{n^2}\to0.
\]
Also $\bar\psi=2R=o_P(1)$. To find its quadratic limit, expand
\[
 \frac1n\sum_i\psi_i^2
 =\frac1n\sum_i(u^\top V_i)^2+\frac1n\sum_i(v^\top U_i)^2
 +\frac2n\sum_i(u^\top V_i)(v^\top U_i).
\]
Since $u$ is independent of the entire $V$ matrix, the first average differs from $\|u\|^2$ by a term whose conditional variance is $2\|u\|^4/n$. Its expectation is bounded when $p_n/n$ is bounded, while $\|u\|^2\to_P\kappaG$. The second average has the same limit.

For the cross term, write $U=\mathbf1u^\top+U_c$ and $V=\mathbf1v^\top+V_c$. Gaussian sample means are independent of their centered matrices. Conditional on $(u,v)$, the cross average is
\[
 R^2+\frac1n(V_cu)^\top(U_cv),
\]
where the second term has mean zero and variance $(n-1)\|u\|^2\|v\|^2/n^2=O_P(n^{-1})$. Thus it vanishes. We obtain $n^{-1}\sum_i\psi_i^2\xrightarrow{P}2\kappaG$. Independence of $X$ from $(U,V)$ and conditional variance $\sigma^2 n^{-2}\sum_i\psi_i^2$ give $n^{-1}\sum_iX_i\psi_i\xrightarrow{P}0$.

The exact jackknife decomposition in S5.5 now yields $n\VarJK\to_P\sigma^2+2\kappaG$, whereas the actual variance is $(\sigma^2+p_n/n)/n$. Assumption~5.1(ii) is not satisfied: the deletion quadratic limit is $2\kappaG$, not the remainder variance limit $\kappaG$. In this Gaussian model the explicit estimator $\widehat V_{\mathrm{analytic}}=s_X^2/n+p_n/n^2$ is consistent and gives valid Wald inference. The analytic correction uses the known dimension $p_n$ of the Gaussian mean model. \qed

\begin{table}[p]
\centering
\scriptsize
\caption{Gaussian boundary validation. Variances are multiplied by $n$; CP denotes coverage.}
\label{tab:arxiv_s15}
\setlength{\tabcolsep}{3pt}
\resizebox{\linewidth}{!}{%
\begin{tabular}{rrrrrrrrr}
\toprule
$n$ & $p$ & MC var. & True var. & JK var. & JK limit & CP oracle & CP JK & CP analytic \\
\midrule
100 & 10 & 1.119 & 1.100 & 1.197 & 1.200 & 0.931 & 0.951 & 0.942 \\
100 & 50 & 1.490 & 1.500 & 2.018 & 2.000 & 0.889 & 0.975 & 0.949 \\
400 & 40 & 1.112 & 1.100 & 1.199 & 1.200 & 0.937 & 0.955 & 0.948 \\
400 & 200 & 1.504 & 1.500 & 2.002 & 2.000 & 0.888 & 0.976 & 0.949 \\
\bottomrule
\end{tabular}
}
\end{table}

\begin{table}[p]
\centering
\scriptsize
\caption{Model (M), 2,000 replicates per cell. All variance and covariance columns are multiplied by $J$. Constants are computed independently by enumerating crossover sequences.}
\label{tab:arxiv_s16}
\setlength{\tabcolsep}{3pt}
\resizebox{\linewidth}{!}{%
\begin{tabular}{rrrrrr}
\toprule
$J$ & Base variance & Own variance & Own sandwich & Residual var. & Twice covariance \\
\midrule
30 & 0.1371 & 0.0648 & 0.0080 & 0.0246 & -0.0883 \\
100 & 0.1217 & 0.0617 & 0.0077 & 0.0210 & -0.0790 \\
400 & 0.1217 & 0.0585 & 0.0076 & 0.0217 & -0.0844 \\
\bottomrule
\end{tabular}
}
\end{table}

\end{document}